\documentclass[a4paper,onecolumn,unpublished]{quantumarticle}
\pdfoutput=1
\usepackage[utf8]{inputenc}
\usepackage[english]{babel}
\usepackage[T1]{fontenc}
\usepackage{amsmath}
\usepackage{amssymb}
\usepackage{physics}
\usepackage{xcolor}
\usepackage{mathtools}
\usepackage{hyperref}
\usepackage{parskip}

\usepackage{amsthm}
\newtheorem{thm}{Theorem}
\newtheorem{lem}{Lemma}
\newtheorem{prop}{Proposition}

\theoremstyle{definition}
\newtheorem{defi}{Definition}

\renewcommand{\aa}[2]{\ensuremath{a_{#1}^{#2}}}

\makeatletter
\newsavebox\myboxA
\newsavebox\myboxB
\newlength\mylenA

\newcommand*\xoverline[2][0.75]{%
    \sbox{\myboxA}{$\m@th#2$}%
    \setbox\myboxB\null% Phantom box
    \ht\myboxB=\ht\myboxA%
    \dp\myboxB=\dp\myboxA%
    \wd\myboxB=#1\wd\myboxA% Scale phantom
    \sbox\myboxB{$\m@th\overline{\copy\myboxB}$}%  Overlined phantom
    \setlength\mylenA{\the\wd\myboxA}%   calc width diff
    \addtolength\mylenA{-\the\wd\myboxB}%
    \ifdim\wd\myboxB<\wd\myboxA%
       \rlap{\hskip 0.5\mylenA\usebox\myboxB}{\usebox\myboxA}%
    \else
        \hskip -0.5\mylenA\rlap{\usebox\myboxA}{\hskip 0.5\mylenA\usebox\myboxB}%
    \fi}
\makeatother
\begin{document}
\title{Gaussian decomposition of magic states for matchgate computations}
\author{Joshua Cudby}
\affiliation{DAMTP, Centre for Mathematical Sciences, University of Cambridge, Cambridge CB30WA, UK}
\author{Sergii Strelchuk}
\affiliation{DAMTP, Centre for Mathematical Sciences, University of Cambridge, Cambridge CB30WA, UK}
\affiliation{Warwick Quantum, Department of Computer Science, University of Warwick, UK}
%%%%%%%%%%%%%%%%%%%%%%%%%%%%%%%%%%%%%%
\begin{abstract}
Magic states, pivotal for universal quantum computation via classically simulable Clifford gates, often undergo decomposition into resourceless stabilizer states, facilitating simulation through classical means. 
This approach yields three operationally significant metrics: stabilizer rank, fidelity, and extent. 
We extend these simulation methods to encompass matchgate circuits (MGCs), and define equivalent metrics for this setting.
We begin with an investigation into the algebraic constraints defining Gaussian states, marking the first explicit characterisation of these states. 
The explicit description of Gaussian states is pivotal to our methods for tackling all the simulation tasks.
Central to our inquiry is the concept of Gaussian rank -- a pivotal metric defining the minimum terms required for decomposing a quantum state into Gaussian constituents. This metric holds paramount significance in determining the runtime of rank-based simulations for MGCs featuring magic state inputs. The absence of low-rank decompositions presents a computational hurdle, thereby prompting a deeper examination of fermionic magic states.
We find that the Gaussian rank of 2 instances of our canonical magic state is 4 under symmetry-restricted decompositions. Additionally, our numerical analysis suggests the absence of low-rank decompositions for 2 or 3 copies of this magic state.
Further, we explore the Gaussian extent, a convex metric offering an upper bound on the rank. We prove the Gaussian extent's multiplicative behaviour on 4-qubit systems, along with initial strides towards proving its sub-multiplicative nature in general settings.
One important result in that direction we present is an upper bound on the Gaussian fidelity of generic states.
\end{abstract}
\maketitle
\tableofcontents
%%%%%%%%%%%%%%%%%%%%%%%%%%%%%%%%%%%%%%
\section{Introduction}
The relation between quantum and classical computational power is one of the most intriguing questions in quantum information science. 
Finding the minimum extra resources that could lift classically efficiently simulable quantum systems to performing universal quantum computation has been a particularly fruitful avenue~\cite{jozsa2013classical}. 
It motivates the study of classical simulation algorithms for quantum systems which allows identification of the regimes where quantum computing does not offer an advantage~\cite{gottesman1998heisenberg,valiant2001quantum}.
For such systems, a lack of quantum advantage means that any quantum process may be simulated classically with at most a polynomial time overhead.

The development of classical simulation algorithms has many benefits beyond outlining the limits of classical computations. 
By simulating quantum systems classically, one can study and probe the properties of quantum systems in the absence of a full-scale quantum computer.
Furthermore, classical simulation is pivotal for the validation of quantum devices.
It provides a way to check the accuracy of noisy quantum computers and other experimental quantum systems; matching classical simulation results to experimental device outputs helps validate that the devices work correctly.
As quantum devices improve, the limits of classical simulation also provide a valuable benchmark that a quantum computer must outperform to surpass classical supercomputers.
Performing a task that would be prohibitively expensive on a classical computer, taking years or centuries say, would achieve the lofty goal of quantum supremacy~\cite{arute_quantum_2019, huang2020classical, kim_evidence_2023}.

One remarkable class of quantum computations with deep connections to fermionic linear optical systems is realised by matchgate circuits (MGCs)~\cite{Terhal2002, knill2001fermionic, jozsa2008matchgates,valiant2001quantum}. It is known that circuits composed solely of matchgates can be efficiently simulated classically~\cite{valiant2001quantum, Terhal2002, jozsa2008matchgates}. 
More precisely, MGCs consist of nearest-neighbour 2-qubit gates of the form:
\begin{equation}
    \begin{gathered}
    G(A, \, B) = \begin{pmatrix}
        p & 0 & 0 & q \\
        0 & w & x & 0 \\
        0 & y & z & 0 \\
        r & 0 & 0 & s
    \end{pmatrix}
    \qquad
    A = \begin{pmatrix}
        p & q \\
        r & s
    \end{pmatrix}
    \quad
    B = \begin{pmatrix}
        w & x \\
        y & z
    \end{pmatrix}  
    \end{gathered}
    \label{eq:matchgate}
\end{equation}
where $A,\, B \in U(2)$ are such that $\det(A) = \det(B)$, and nearest-neighbour refers to some fixed linear ordering of qubits.
One can consider $A$ (\textit{resp.} $B$) as acting on the even (odd) Hamming weight subspace of the two input qubits, with no mixing between the subspaces.
As such, MGCs preserve the parity of the Hamming weight of input states, leading to a partition of the set of outputs of MGCs into even and odd parts. 

We will refer to the output of MGCs acting on a fixed computational-basis input state with even (\emph{resp.} odd) Hamming weight as even (odd) Gaussian states. 
In the literature, these are often referred to as fermionic Gaussian states, to distinguish them from bosonic Gaussian states. Here, we consider only the fermionic setting and do not make the distinction explicit.

MGCs describe the computational ability of unassisted fermionic linear optics; that is, they correspond to a model of non-interacting fermions in 1D. \cite{Terhal2002}. 
As such, MGCs and Gaussian states are closely related to the theory of Majorana fermions, which we briefly introduce here.

In the standard theory of $n$ fermionic modes, we consider the creation and annihilation operators $a_i,\,a_i^\dag$ for $i = 1,\,\ldots,\,n$.
These satisfy the standard anti-commutation relations:
\begin{align}
    \{a_i,\,a_j\} = \{a_i^\dag,\,a_j^\dag\} = 0 \quad \{a_i,\,a_j^\dag\} = \delta_{ij} I
\end{align}

One can alternatively consider the Majorana fermions $c_i$ for $i = 1,\,\ldots,\,2n$. These are related via
\begin{align}
    a_i &= \frac{1}{2}(c_{2i - 1} + i c_{2i}) \\
    a_i^\dag &= \frac{1}{2}(c_{2i - 1} - i c_{2i})
\end{align}
One can easily check that the Majorana fermions obey the anti-commutation relations:
\begin{equation}
    \{c_i,\,c_j\} = 2 \delta_{ij} I \label{eq:clifford_alg}
\end{equation}
\autoref{eq:clifford_alg} defines a Clifford algebra on $2n$ generators.
One possible representation of the algebra is the Jordan-Wigner representation~\cite{jordan_uber_1928}, under which the Majorana fermions are mapped to Pauli operators on $n$ qubits:
\begin{equation}
    c_{2k - 1} = \left(\prod_{i = 1}^{k - 1} Z_i\right) X_k \quad
    c_{2k} = \left(\prod_{i = 1}^{k - 1} Z_i\right) Y_k
\end{equation}

The link between Gaussian states and Majorana fermions is most clearly seen in the following characterisation of Gaussian states. One can show that a state $\ket{\psi}$ on $n$ qubits is Gaussian if and only if~\cite{Bravyi2004}:
\begin{equation}
    \Lambda \ket{\psi}\ket{\psi} = \sum_{k = 1}^{2n} c_k \otimes c_k \ket{\psi}\ket{\psi} = 0  \label{eq:gaussian_state_eqn}
\end{equation}
This has historically been the primary operational description of Gaussian states.

While this description suffices for MGCs acting on computational basis inputs, it is well known that a description of a gate set is not enough to specify the power of quantum processes.
The presence of other resources in a quantum process can drastically impact computational power. 
One such resource is a broader class of allowed input states, with so-called \textit{magic states} being those inputs that elevate a process to full universal quantum computation \cite{Bravyi2005}. 

It has been shown that every pure fermionic state that is non-Gaussian (that is, cannot be generated by an MGC from a computational basis state), is a magic state for MGCs~\cite{Hebenstreit2019}. The notion of magic states for MGCs turns out to be rather more nuanced when contrasted with another well-known class of quantum computations -- Clifford computations -- and their associated magic states~\cite{Bravyi2019, Labib_2022, lovitz2022new, mehraban2023quadratic, Peleg_2022}. This is because the locality of interaction plays a significant role for MGCs~\cite{jozsa2008matchgates, brod2011extending}.

Classically simulating Clifford circuits or MGCs with a large number of magic state inputs is inefficient in general. 
One simulation technique involves decomposing the magic state inputs into sums of resourceless states,  each of which is obtainable by the corresponding gate sets~\cite{Bravyi_2016, Bravyi_2016_improved, Bravyi2019}.
The quantum process on any term in this sum is classically simulable by the celebrated Gottesman-Knill theorem~\cite{gottesman1998heisenberg} or the results of Valiant~\cite{valiant2001quantum} respectively for the two settings. 
By linearity, the output of the circuit can be reconstructed from the classical simulation of each term in the sum. 
Simulation time can be dramatically reduced if one can find a decomposition of these magic states into a smaller number of terms. 

A key quantity that determines the run time of the above class of algorithms for the case of Clifford circuits is the stabilizer rank~\cite{Bravyi_2016}, which is the minimal number of terms needed in a decomposition into stabilizer states. 
Bravyi et al~\cite{Bravyi2019} also define the associated stabilizer extent $\xi$ and fidelity $F_C$ as follows:
\begin{align}
    \xi(\ket{\Phi}) &= \min \norm{c}_1^2 \quad \text{s.t. } 
    \sum_{i = 1}^k c_i \ket{\psi_i} = \ket{\Phi} 
    \\
    F_C(\ket{\Phi})& = \max_i \abs{\braket{\psi_i}{\Phi}}^2
\end{align}
where the $\ket{\psi_i}$ are stabilizer states. It is known that the stabilizer extent is not multiplicative \cite{Heimendahl2021}, and as part of the proof it was shown that generic Haar-random states have exponentially low fidelity.

Our results are organised as follows. 
In~\autoref{sec:notation}, we define our notation and some preliminaries.
In~\autoref{sec:properties}, we give an explicit description of Gaussian states and use this characterisation to prove results on the dimension of the Gaussian manifold and on sums of 2 Gaussian states. 
In~\autoref{sec:fidelity}, we construct an $\epsilon$-net for the Gaussian manifold and use it to upper bound the Gaussian fidelity of generic Haar-random states. 
In~\autoref{sec:extent}, we discuss the multiplicativity of the Gaussian extent~\footnote{While preparing this manuscript we became aware of related results on Gaussian extent by  Oliver Reardon-Smith, Kamil Korzekwa, Michal Oszmaniec (TQC 2023, \cite{korzekwa2023simulation}) and independently by Robert Koenig, Beatriz Cardoso Diaz \cite{koenig2023simulation}}. 
In~\autoref{sec:rank} we discuss our numerical attempt to find low-rank Gaussian decompositions for 2 or 3 copies of a certain Gaussian magic state. 
We also give a proof that no such decomposition exists for 2 copies in a symmetry-reduced case.
%%%%%%%%%%%%%%%%%%%%%%%%%%%%%%%%%%%%%%
\section{Notation and Preliminaries} \label{sec:notation}
We use $\Re$ to denote the real part of an expression.

We will generally make no distinction between an integer $x$ and its binary representation. 
Bitwise XOR of binary vectors will be notated as $\oplus$.
The Hamming weight of a binary string is $\abs{\cdot}$ and the Hamming distance is $d(\cdot,\,\cdot)$.

Let $D(x,\,y) \coloneqq \{ k_i : x_{k_i} \neq y_{k_i}; \ i = 1,\,\ldots,\,d(x,\,y)\}$, where the $k_i$ are written in increasing order.

We refer to the substring of a binary string, running from indices $i$ to $j$ inclusive, as $x_{1:j}$. 

We use an overline to denote changing bits of a binary string: $\overline{x}^i = x \oplus e_i$ and $\overline{x}^{i,\,j} = x \oplus e_i \oplus e_j$.

We denote the set of all binary strings of length $n$ by $\mathcal{B}_n$ and the subset of even-weight strings by $\mathcal{A}_n$. The odd-weight strings are then the complement of this set, denoted $\mathcal{A}_n^\mathsf{c}$.

We consider circuits of gates on a linear ordering of qubits. 
A matchgate is a 2-qubit gate acting on nearest neighbours, whose matrix is given by~\autoref{eq:matchgate}.
A matchgate circuit (MGC) is a circuit composed entirely of matchgates.
The set of Gaussian operations, i.e. those operations corresponding to matchgate circuits, is $\mathcal{G}_n \subset \mathcal{U}_n$.

We say that a state $\ket{\psi}$ is of (definite) even parity if it may be written in the computational basis as a sum over only labels with even Hamming weight: $\ket{\psi} = \sum_{x \in \mathcal{A}_n} a_x \ket{x}$. Similarly, an odd parity state may be written $\ket{\psi} = \sum_{x \in \mathcal{A}_n^\textsf{c}} a_x \ket{x}$.

We refer to the Hilbert space of states with even parity by $\mathcal{H}_n$.

An even Gaussian state on $n$ qubits is any state that can be obtained from an MGC acting on the input state $\ket{0}^{\otimes n}$.
Since MGCs preserve the parity of inputs, these states are of definite even parity.
We denote the set of even Gaussian states on $n$ qubits by $G_n$: 
$$G_n \coloneqq \{ \ket{s} = U\ket{0}^{\otimes n} : U \in \mathcal{G}_n\}$$
We will drop the subscript when it is obvious by context. 
%%%%%%%%%%%%%%%%%%%%%%%%%%%%%%%%%%%%%%
\section{Properties of Gaussian states} \label{sec:properties}

\autoref{eq:gaussian_state_eqn} provides a complete characterisation of Gaussian states; however, it is not easy to work with in practice.
Starting \textit{a priori} from this equation, it is not immediately clear how the amplitudes of a Gaussian state are constrained and interlinked.

We derive an explicit representation of Gaussian states by finding an independent set of the constraints imposed by~\autoref{eq:gaussian_state_eqn}.
This representation is vital for the derivation of the rest of our results.

\subsection{Independent constraints defining Gaussian states}
Recall the definition of a Gaussian state:
\begin{equation}
    \Lambda \ket{\psi}\ket{\psi} = \sum_{k = 1}^{2n} c_k \otimes c_k \ket{\psi}\ket{\psi} = 0  \label{eq:gaussian_state_eqn_2}
\end{equation}
where the $c_i$ are given by:
\begin{equation}
    c_{2k - 1} = \left(\prod_{i = 1}^{k - 1} Z_i\right) X_k \quad
    c_{2k} = \left(\prod_{i = 1}^{k - 1} Z_i\right) Y_k
\end{equation}
On a computational basis state $\ket{x}$, we have:
\begin{align}
    c_{2k-1}\ket{x} &= (-1)^{\abs{x_{1:k}}} \ket{\overline{x}^k} \\
    c_{2k}\ket{x} &= i(-1)^{\abs{x_{1:k}}} (-1)^{x_k} \ket{\overline{x}^k}
\end{align}
In both cases, we get a sign factor from the string of $Z_i$, $i = 1,\,\ldots,\,k-1$, and we flip the $k$th bit. For $c_{2k}$, we also pick up a factor of $i$ and an additional $(-1)^{x_k}$.

Let $\ket{\psi} = \sum_{x \in \mathcal{A}_n} a_x \ket{x}$ be an even-parity state.
From~\autoref{eq:gaussian_state_eqn_2}, we can see it is Gaussian if and only if:
\begin{align}
    \Lambda \ket{\psi}\ket{\psi} &= 
    \sum_{k = 1}^{2n} \bigl(c_k \otimes c_k \bigr) \sum_{x, y \in \mathcal{A}_n} a_x a_y \ket{x} \ket{y} \nonumber \\
    &= \sum_{x, y \in \mathcal{A}_n} a_x a_y \sum_{k = 1}^{2n} \bigl(c_k \otimes c_k \bigr)\ket{x} \ket{y} \nonumber \\
    &= \sum_{x, y \in \mathcal{A}_n} \left( a_x a_y \sum_{k = 1}^n \biggl( (-1)^{\abs{x_{1:k-1}} + \abs{y_{1:k-1}}} \bigl(1 + (i^2)(-1)^{x_k + y_k} \bigr)  \ket{\xoverline{x}^k}\ket{\xoverline{y}^k} \biggr) \right) \nonumber\\
    &= \sum_{x, y \in \mathcal{A}_n^\mathsf{c}} \left( \sum_{k = 1}^n \biggl( (-1)^{\abs{x_{1:k-1}} + \abs{y_{1:k-1}}} 
    \bigl(1 - (-1)^{x_k + y_k} \bigr) a_{\overline{x}^k} a_{\overline{y}^k} \biggr) \ket{x}\ket{y} \right) \label{eq:expanded_lambda_eq} = 0 
\end{align}

Fix $x,\,y \in {\mathcal{A}_n^\mathsf{c}}^{\otimes 2}$, and consider the inner sum of~\autoref{eq:expanded_lambda_eq}:
\begin{equation}
    \sum_{k = 1}^n (-1)^{\abs{x_{1:k-1}} + \abs{y_{1:k-1}}} 
    \bigl(1 - (-1)^{x_k + y_k} \bigr) a_{\overline{x}^k} a_{\overline{y}^k} \label{eq:inner_sum}
\end{equation}
If $x_k = y_k$, then the summand in~\autoref{eq:inner_sum} vanishes. So we may restrict the sum to only indices $k$ where $x_k \neq y_k$, i.e. to the set $D(x,\,y)$.

We now show that terms that survive in~\autoref{eq:inner_sum} alternate sign.
Let $k,\,k'$ be a pair of consecutive indices in $D(x,\,y)$. 
Then we may consider the relative sign between the terms corresponding to $k,\,k'$, which is given by:
$$
(-1)^{\abs{x_{1:k'-1}} + \abs{y_{1:k'-1}} - \abs{x_{1:k-1}} - \abs{y_{1:k-1}}} 
= (-1)^{\abs{x_{k:k'-1}} + \abs{y_{k:k'-1}}}
$$
Since, by construction, $k+1,\,\ldots,\,k'-1 \notin D(x,\,y)$, we have $x = y$ on this substring.
These bits therefore contribute no phase.
Then, since $x_k \neq y_k$, we have that the relative phase is $-1$. 

We may therefore simplify the form of~\autoref{eq:inner_sum}.
For any $x,\,y \in {\mathcal{A}_n^\mathsf{c}}^{\otimes 2}$, we have 
\begin{equation}
    \sum_{i = 1}^{\abs{D(x,\,y)}} (-1)^{i+1} a_{\xoverline{x}^{k_i}}a_{\xoverline{y}^{k_i}}
     \label{eq:constraint}
\end{equation}
It will be useful to write these equations in terms of labels with even parity rather than odd.
To do this, we treat the first term of the sum in~\autoref{eq:constraint} as preferential. 
Let $u = \overline{x}^{k_1}$ and $v = \overline{y}^{k_1}$, which have even weight.
Then we may re-write~\autoref{eq:constraint} as:
\begin{equation} \label{eq:f_uv}
    f(u,\,v) \coloneqq a_u a_v - \sum_{i = 2}^{\abs{D(u,\,v)}} (-1)^{i} a_{\overline{u}^{k_1,\,k_i}}a_{\overline{v}^{k_1,\,k_i}}
\end{equation}
We will think of the equation $f(u,\,v) = 0$ as the constraint associated with labels $u,\,v \in \mathcal{A}_n$.
Note the abuse of notation to avoid excessive sub- or super-scripts in the definition of the constraints.
For any $u,\,v \in \mathcal{A}_n$, $f(u,\,v)$ is defined on the set of amplitudes $\{a_x : x \in \mathcal{A}_n \}$ of an even parity state $\ket{\psi}$. Which state the function $f(u,\,v)$ is defined on will always be clear from context.

Note that if $d(u,\,v) = 2$ then the two terms on the RHS of~\autoref{eq:f_uv} cancel and the constraint $f(u,\,v) = 0$ is trivial.
$f(u,\,u) = 0$ is trivial also, since the corresponding sum in~\autoref{eq:constraint} has no terms. 

Therefore, \autoref{eq:expanded_lambda_eq} is true if and only if the following set of constraints is satisfied:
\begin{equation}
    \mathcal{F} = \left\{
     f(x,\,y) = 0
    : (x,\,y) \in \mathcal{A}_n^{\otimes 2}
    ;\, d(x,\,y) \geq 4
    \right\}
\end{equation}
For example, for $n = 4$, the only non-trivial constraint in $\mathcal{F}$ (up to permutation of terms) is
\begin{equation}
    f(0000,\,1111) \equiv f(0,\, 15) = a_0 a_{15} - a_3 a_{12} + a_5 a_{10} - a_6 a_9 = 0
\label{eq:f_0_15}
\end{equation}

Ignoring duplication, $\mathcal{F}$ has $2^{n-1}\left( 
2^{n-1} - \tfrac{n(n-1)}{2} -1
\right)$ equations; not all of them are independent.
To find an independent set of equations, we will use a complexified form of the Implicit Function Theorem:
\begin{thm}[Implicit Function Theorem, see e.g. \cite{Hormander1973}] \label{thm:ift}
    Let $\left\{g_j(w, z) : j = 1, 2, \dots, m \right\}$ be analytic functions of $(w, z) = (w_1, \dots, w_m, z_1, \dots, z_n)$ in a neighbourhood of $(w_0, z_0) \in \mathbb{C}^m \cross \mathbb{C}^n$.
    Suppose $g_j(w_0, z_0) = 0 \ \forall j$ and also that 
    $J = \det \left(\pdv{g_j}{w_k}\right)_{j, k = 1}^{m} \neq 0 \ \text{at} \ (w_0, z_0)$.
    Then, in a neighbourhood of $z_0$, $\{g_j(w, z)\}$ has a unique analytic solution such that $w(z_0) = w_0$, given implicitly by $(w,\,z) = (w(z),\,z)$.
\end{thm}

This leads us to our first result:
\begin{prop} \label{prop:constraints}
    An even-parity Gaussian state on $n$ qubits is defined by a set of $2^n - \tfrac{n(n-1)}{2}$ algebraic constraints.
    These may be chosen as $\mathcal{F}^y_{\text{indep}} \coloneqq \{f(y,\,w) = 0 : d(y,\,w) \geq 4\}$, where $f$ is as defined above, along with a normalisation and global phase equivalence condition.
\end{prop}

\begin{proof}
    Assume without loss of generality that $a_0 \neq 0$; the proof easily generalises to any other $y \in \mathcal{A}_n$ such that $a_y \neq 0$.
    
    Consider the following partition of $\mathcal{A}_n$:
    \begin{equation*}
        Z = \left\{ x \in \mathcal{A}_n : \abs{x} \leq 2 \right\} \quad 
        W = \left\{ x \in \mathcal{A}_n : \abs{x} \geq 4 \right\}
    \end{equation*}
    Define a further partition of $W$ based on the first `set' bit:
    \begin{equation*}
        W_k = \left\{
            w \in W : \ w_{1:k-1} = 0 ; \ w_k = 1
        \right\} \quad k = 1,2,\ldots n-3
    \end{equation*}    
    We pick the following set of functions, and show that they fulfil the conditions of the Implicit Function Theorem:
    \begin{equation}
        \bigcup_{k = 1}^{n-3} \left\{
            f(0,\,w) : w \in W_k
        \right\}
        \equiv \{ f(0,\,w) : \abs{w} \geq 4\}
    \end{equation}
    
    Note that, for $w \in W_k$, we have:
    \begin{equation} \label{eq:jacobian_elements}
        \pdv{f(0,\,w)}{a_v} = \begin{cases}
            a_0 \delta_{wv} &\quad v \in W_k \\
            0 &\quad v \in W_l,\, l < k \\
            {\vcenter{\hbox{\tiny$\bullet$}}} &\quad v \in W_l,\, l > k \\
        \end{cases}
    \end{equation}
    
  \autoref{eq:jacobian_elements} gives us that the Jacobian is of block-upper-triangular form, with multiples of the identity on the diagonal. We need not consider the form of $\pdv{f(0,\,w)}{a_v}$ for $v \in W_l,\, l > k$ since block-upper-triangular is sufficient for our determinant calculation.
    
    Letting $m_k = \abs{W_k}$, we have: 
    \begin{equation}
        J = 
        \left(\begin{array}{c|c|c|c}
          a_0 I_{m_1}  & {\vcenter{\hbox{\tiny$\bullet$}}} & {\vcenter{\hbox{\tiny$\bullet$}}} & \dots \\
        \hline
          0 & a_0 I_{m_2} & {\vcenter{\hbox{\tiny$\bullet$}}}  & \dots \\
        \hline
          0 & 0 & a_0 I_{m_3} & \dots \\
        \hline
          \vdots & \vdots & \vdots & \ddots
        \end{array}\right)
    \end{equation}
    
    The determinant of $J$ is therefore just the product of the determinants of the diagonal blocks and is non-zero, since $a_0 \neq 0$ by assumption.

    We now have a set of functions, a partition of variables and a non-zero determinant, as required by \autoref{thm:ift}.
    
    \autoref{thm:ift} therefore implies that there is a unique analytic solution to the set of constraints $\mathcal{F}_{\text{indep}}^0 \coloneqq \{f(0,\,w) = 0 : \abs{w} \geq 4\}$. 
    For this solution, for any $w \in W$, we have an implicit expression $a_w = a_w(\{a_z : z \in Z\})$.
    
    In words, the amplitudes with labels $w,\,\abs{w} \geq 4$ are functions of the amplitudes with labels $z,\,\abs{z} \leq 2$.
    
    Since $\abs{Z} = \binom{n}{2} + 1$, we have an independent set of $2^n - \tfrac{n(n-1)}{2}$ constraints as claimed.

    The proof generalises to any favoured $y \in \mathcal{A}_n$ with $a_y \neq 0$ by replacing the bit string 0 by $y$ and the Hamming weight $\abs{\cdot}$ by $d(y,\,\cdot)$ throughout.

    The corresponding independent set of constraints which define Gaussian states is then:
    \begin{equation} \label{eq:indep_constraints_y}
        \mathcal{F}_{\text{indep}}^y = \left\{ f(y,\,w) = a_y a_w - \sum_{i = 2}^{|D(w,\,y)|}(-1)^i a_{\overline{y}^{k_1,k_i} } a_{\overline{w}^{k_1,k_i} } = 0 : d(y,\,w) \geq 4 \right\}
    \end{equation}
\end{proof}
\subsection{Triples of Gaussian states}
The explicit representation above is useful for computing further properties of Gaussian states. For example, we can quantify the conditions for 2 Gaussians to sum to a third:
\begin{prop}
    Suppose $\ket{\psi_0} = \alpha \ket{\psi_1} + \beta \ket{\psi_2}$ where the $\ket{\psi_i}$ are normalised Gaussian states and without loss of generality take $\alpha,\, \beta \in \mathbb{R}$. Then the triple can be expressed as
    \begin{equation}
         \{\ket{\psi_0},\,  \ket{\psi_1}, \ket{\psi_2} \} = \left\{
        U\ket{0}, \, 
        \frac{1}{\alpha} U\left(
        \sum_{\substack{y \in \mathcal{A}_n :\\ \abs{y} \leq 2}} a_y \ket{y}
        \right), \, 
        \frac{1}{\beta}U \left(
            (1-a_0)\ket{0} - \sum_{\substack{y \in \mathcal{A}_n :\\ \abs{y} = 2}} a_y \ket{y}
        \right) 
        \right\}
    \end{equation}
    for some Gaussian operation $U \in \mathcal{G}$. 
    Here, the $\{a_y : \abs{y} \leq 2 \}$ are some complex coefficients such that 
    $$\sum_{\substack{y \in \mathcal{A}_n :\\ \abs{y} \leq 2}} a_y \ket{y} 
    \quad \text{and} \quad
    (1-a_0)\ket{0} - \sum_{\substack{y \in \mathcal{A}_n :\\ \abs{y} = 2}} a_y \ket{y}$$ 
    are Gaussian states.
    
    Moreover, for any $\ket{\psi}$, the dimension of the manifold of Gaussians for which  $\alpha \ket{\psi} + \beta\ket{\psi'} \in G$ is $2n-3$.
\end{prop}

\begin{proof}
    Let $U \in \mathcal{G}$ be the Gaussian operator s.t. $\ket{\psi_0} = U\ket{0}$.
    Define also 
    $\ket*{\tilde{\psi}_1} = U^\dag\ket{\psi_1}, \, \ket*{\tilde{\psi}_2} = U^\dag\ket{\psi_2}$.
    Then we have $\ket{0} = \alpha \ket*{\tilde{\psi_1}} + \beta \ket*{\tilde{\psi_2}}$.

    Write $\alpha \ket*{\tilde{\psi}_1} = \sum_{x\in \mathcal{A}_n} a_x \ket{x}$ where without loss of generality we can set $0 < a_0 \leq 1$. Then $\beta\ket*{\tilde{\psi}_2} = (1- a_0)\ket{0} - \sum_{x \in \mathcal{A}_n : \abs{x} \geq 2} a_x \ket{x}$.

    We require these states to be Gaussian. 
    For each $x \in \mathcal{A}_n$ with $\abs{x} \geq 4$, we have a constraint of the form:
    \begin{equation} \label{eq:triple_constraint_general}
        f(0,\,x) = 
        c_0 c_x - \sum_{i = 2}^{\abs{x}} (-1)^i c_{\overline{0}^{k_1,k_i}} c_{\overline{x}^{k_1,k_i}} = 0
    \end{equation}

    Imposing this constraint for $\alpha \ket*{\tilde{\psi}_1}$ and $\beta\ket*{\tilde{\psi}_2}$ respectively gives:
    \begin{align}
        a_0 a_x - \sum_{i = 2}^{\abs{x}} 
        (-1)^i a_{\overline{0}^{k_1,k_i}} a_{\overline{x}^{k_1,k_i}} &= 0 \label{eq:triples_constraints_1}\\
        (1-a_0)(-a_x) - \sum_{i = 2}^{\abs{x}} (-1)^i 
        (-a_{\overline{0}^{k_1,k_i}} ) (-a_{\overline{x}^{k_1,k_i}} ) &= 0
        \label{eq:triples_constraints_2}
    \end{align}
    Comparing \autoref{eq:triples_constraints_1} \& \autoref{eq:triples_constraints_2} we see that 
    $a_x = 0$ for all $\abs{x} \geq 4$.

    So we have 
    $\alpha \ket*{\tilde{\psi}_1} = a_0 \ket{0} + \sum_{x\in \mathcal{A}_n : \abs{x} = 2} a_x \ket{x}$, and 
    $\beta\ket*{\tilde{\psi}_2} = (1- a_0)\ket{0} - \sum_{x \in \mathcal{A}_n : \abs{x} = 2} a_x\ket{x}$.
    This gives us the first part of our result.

    For the second part, it is necessary to pick the $a_y$ with $\abs{y} \leq 2$ in such a way that $a_x = 0$ for $\abs{x} \geq 4$. 
    
    Note that $a_x = 0$ for all $\abs{x} = 4$ forces $a_x = 0$ for all $\abs{x} \geq 4$ due to the recursive nature of~\autoref{eq:triple_constraint_general}. 
    
    For $\abs{x} = 6$, since $\abs{\overline{x}^{k_1,\,k_i}} = 4$, upon imposing $a_{x'} = 0$ for all $\abs{x'} = 4$, \autoref{eq:triple_constraint_general} becomes $f(x) = a_0 a_x = 0$.
    Thus $a_x = 0$.
    
    Proceeding iteratively for $\abs{x} = 8,\,10,\,\ldots$, we see that indeed $a_x = 0$ for all $\abs{x} \geq 4$.
    
    The relevant equations, which act as constraints on the $a_y$ with $\abs{y} \leq 2$, are then:
    \begin{equation*}
        \tilde{\mathcal{F}} = \left\{
        \sum_{i = 2}^{4} (-1)^i a_{\overline{0}^{k_1,k_i}} a_{\overline{y}^{k_1,k_i}} = 0 
        : \abs{x} = 4,\, D(0,\,x) = \{k_i : i = 1,\,\ldots,\,4\}
        \right\}
    \end{equation*}
    For each $x, \, \exists \, \binom{n-2}{2}$ choices of $y \in \mathcal{A}_n$ with $\abs{y} = 4$ st $d(x, \, y) = 2$.
    In other words, each $a_x$ appears in exactly $\binom{n-2}{2}$ equations. 
    We may assume $a_x \neq 0$ for some $\abs{x} = 2$, else the vectors in the linearly dependent triple are all parallel. 
    In~\autoref{sec:app.triple} we show that satisfying the equations explicitly containing $a_x$ is sufficient to satisfy all the equations. 
    This leaves $ \binom{n}{2} - \binom{n-2}{2} = 2n - 3$ of the variables unconstrained.
\end{proof}
%%%%%%%%%%%%%%%%%%%%%%%%%%%%%%%%%%%%%%
\section{An upper bound for Gaussian fidelity} \label{sec:fidelity}
The \textit{Stabilizer Fidelity} was first defined in \cite{Bravyi2019}.
It appeared as a lower bound for the Stabilizer Extent, which we briefly discuss in~\autoref{sec:extent}, and measures the maximum overlap of a state $\ket{\psi}$ with any stabilizer state $\ket{\phi}$.
One can define the \textit{Gaussian Fidelity} similarly.

\begin{defi}
    The \textit{Gaussian Fidelity} is given by $F_G(\ket{\Psi}) \coloneqq \sup_{\ket{s} \in G} \abs{\braket{\Psi}{s}}^2$.
\end{defi}
We will generally refer to the Gaussian fidelity simply as fidelity for brevity. It provides a convenient measure of the closeness of the given state to being Gaussian.
Via the same convex duality arguments as for the stabilizer case, the fidelity naturally gives a lower bound for the Gaussian extent, which is discussed in~\autoref{sec:extent}.

Here, we give an upper bound for the fidelity of generic Haar-random states. 
We begin by constructing an $\epsilon$-net \cite{Hayden2004} for the Gaussian states. 
By a standard union-bound argument, we obtain a bound for the overlap of random states with this discrete set. 
This directly leads to a bound for the overlap with any Gaussian state i.e. a bound for the fidelity.

\subsection{Using a net to bound the fidelity}
\begin{defi}
    An $\epsilon$-net $\mathcal{N}$ is a set of Gaussian states such that for every $\ket{s} \in G$ there exists $\ket*{\tilde{s}} \in \mathcal{N}$ with $\norm{\ket{s}\bra{s} - \ket*{\tilde{s}}\bra*{\tilde{s}}}_1 < \epsilon$ \cite{Hayden2004}. 
\end{defi}

\begin{prop}\label{prop:eps_net}
    Let $\epsilon = 2^{-l}$. Then there is an $\epsilon$-net $\mathcal{N}$ with cardinality $\log_2 \abs{\mathcal{N}} \leq n^4 + 2n^2 + n+ ln^2$.
\end{prop}
We defer the proof to~\autoref{sec:app.net}.
\begin{prop}
    For a generic Haar-random state $\ket{\Psi}$, the Gaussian Fidelity is exponentially small with probability exponentially close to unity. 
    In particular, for any $\delta > 0$,
    $F_G(\ket{\Psi}) \leq 2^{-n + 2}(1+ \delta)n^4$ with probability $\approx 1 - e^{-\delta n^4}$.
\end{prop}
\begin{proof}
    Following the proof of Claim 2 in \cite{Bravyi2019}, for any $n$-qubit state $\ket*{\tilde{s}}$, we have 
    $$\mathbb{P}\left(\abs{\braket*{\tilde{s}}{\Psi}}^2 \geq x\right) = (1-x)^{2^n-1} \leq e^{-x(2^n-1)}$$
    Taking a union bound over $\mathcal{N}$ with $x = 2^{-l}$,
    \begin{align*}
        \mathbb{P}\left(
        \max_{\ket*{\tilde{s}} \in \mathcal{N}} \abs{\braket*{\tilde{s}}{\Psi}}^2 \geq 2^{-l}
        \right) 
        &\leq \abs{\mathcal{N}} \cdot \exp({-2^{-l}(2^n - 1)}) \\
        & \leq \exp(-2^{n-l} + (n^4 + 2n^2 + n + ln^2)\log(2))
    \end{align*}
    Take $l = n - \log_2((1+\delta)n^4)$ for any $\delta > 0$. Then
    $$
    \mathbb{P}\left(
        \max_{\ket*{\tilde{s}} \in \mathcal{N}} \abs{\braket*{\tilde{s}}{\Psi}}^2 \geq 2^{-l}
    \right) \leq \exp(-\delta n^4 + \mathcal{O}(n^2))
    $$
    By definition of a $2^{-l}$-net, we can write $\ket{s} = \ket*{\tilde{s}} + 2^{-l}\ket{\zeta}$, where $\norm{\ket{\zeta}}_1 \leq 1$, $\ket{s} \in G$ and $\ket*{\tilde{s}} \in \mathcal{N}$. Then
    \begin{align*}
        F_G(\ket{\Psi}) &= \sup_{\ket{s} \in G} \abs{\braket{s}{\Psi}}^2 \\
        &= \abs{\braket*{\tilde{s}}{\Psi} + \epsilon \braket{\zeta}{\Psi}} ^ 2 \\
        &\leq (2^{-\frac{l}{2}} + 2^{-l}\cdot 1)^2 \\
        & \leq 2^{-l + 2}\\
        &= 2^{-n + 2}(1+ \delta)n^4
    \end{align*}
    with probability $\approx 1- e^{-\delta n^4}$.
\end{proof}
%%%%%%%%%%%%%%%%%%%%%%%%%%%%%%%%%%%%%%
\section{Multiplicativity of Gaussian extent} \label{sec:extent}

The \textit{Stabilizer Extent} was first defined in \cite{Bravyi2019} as a quantity that appears in an upper bound for the \textit{Approximate Stabilizer Rank}.

For some precision parameter $\delta$, the approximate stabilizer rank $\chi_\delta(\ket{\psi})$ is the smallest integer $k$ such that $\norm{\ket{\psi} - \ket{\psi'}} \leq \delta$ for some state $\ket{\psi'}$ with exact stabilizer rank $k$.

Then the stabilizer extent  $\xi$ is such that:
\begin{equation}
    \chi_\delta(\ket{\psi}) \leq 1 + \frac{\xi(\ket{\psi)}}{\delta^2}
\end{equation}

It was found to be easier to work with than the rank, enjoying a natural phrasing as a Second Order Cone Programming (SOCP) problem. 
For a thorough review of this topic, we refer the reader to e.g. \cite{alizadeh_2001_cone, lobo_applications_1998, Boyd_Vandenberghe_2004}.
Here, we provide a quick summary.

The standard (primal) form of a SOCP is:
\begin{align}
    \min& \ f^T x \\
    \text{s.t.}& \ \norm{A_i x + b_i} \leq c_i^T x + d_i, \quad i = 1,\,\ldots,\,N \nonumber
\end{align}
where $x \in \mathbb{R}^n$ is the optimization variable and $f \in \mathbb{R}^n,\,A_i \in \mathbb{R}^{n_i \cross n},\, b_i \in \mathbb{R}^{n_i},\,c_i \in \mathbb{R}^n,\,d_i \in \mathbb{R}$ are problem parameters.

If $\forall i,\, n_i = 1$, then the SOCP reduces to a Linear Program. 
Similarly, Quadratic Programs may be formulated as a SOCP.
Conversely, any SOCP may be represented by a Semi-Definite Program (SDP).

Various tools from convex optimization then lend themselves to this formulation.
In particular, the standard primal SOCP may be cast into a dual form:
\begin{align}
    \max& \ - \sum_{i=1}^N (b_i^T z_i + d_i w_i) \\
    \text{s.t.}& \ \sum_{i=1}^N (A_i^T z_i + c_i w_i) = f, \nonumber\\
    & \ \norm{z_i} \leq w_i,\quad i = 1,\,\ldots,\,N \nonumber
\end{align}
where the dual optimization variables are $z_i \in \mathbb{R}^{n_i}$ and $w \in \mathbb{R}^N$.

A problem is called \textit{feasible} if there exists $x$ that satisfies all the constraints, and \textit{strictly feasible} if there exists $x$ that satisfies all the constraints with strict inequality.

For a primal-dual pair, let $p^*$ and $d^*$ denote the optimal values of the primal and dual problems respectively. We have the following facts (see e.g. \cite{lobo_applications_1998}):
\begin{enumerate}
    \item (\textit{weak duality}) $p^* \geq d^*$
    \item (\textit{strong duality}) if either problem is strictly feasible, then $p^* = d^*$.
    \item If both problems are strictly feasible, there exist primal and dual feasible points that obtain the optimal values.
\end{enumerate}
The same facts are also true for SDPs.

Returning to the stabilizer extent, Bravyi et al. showed in~\cite{Bravyi2019} that the Stabilizer Extent is multiplicative on systems of at most 3 qubits. 
Conversely, Heimandahl et al. showed in~\cite{Heimendahl2021} that it is \textit{not} in general multiplicative; in fact, they showed that random states give a counterexample with probability exponentially close to 1.

\subsection{Summary of results}
We follow the proof techniques of \cite{Heimendahl2021} almost to their conclusion for the fermionic case in an attempt to prove the strict sub-multiplicativity of the Gaussian extent. 
The proof technique is to consider the dual problem again, finding a pair of points for which the tensor product of their optimal dual witnesses is \emph{not} a feasible point.

We show that, generically, optimal dual witnesses are extreme points of the set of feasible points, and that they are unique. 
These facts are almost immediate for the stabilizer case but require careful treatment in the Gaussian setting.
The facts are necessary to analyse tensor products of primal points and their dual witnesses correctly.
The upper bound on the Gaussian Fidelity of~\autoref{sec:fidelity} implies that the extreme points have an exponentially large norm.
Finding a Gaussian state for which the overlap with $2$ of these extreme points is greater than unity therefore seems plausible.

\subsection{Problem formulation}
For Gaussian states on $n$ qubits, we can formulate the Gaussian extent as a program.
\begin{defi}
    For any $\ket{\Psi} \in \mathcal{H}_n$, the \textit{Gaussian extent} $\xi_G(\ket{\Psi})$ is given by
    \begin{alignat}{3}
        & \sqrt{\xi_G (\ket{\Psi})} = &&\min &&\int_{s \in G} \abs{c(s)} ds \\
        & && \text{s.t.} &&\int_{s \in G} c(s) \ket{s} ds = \ket{\Psi} \nonumber \\
        & &&  &&c(s) \in \mathbb{C} \nonumber
    \end{alignat}
    Since the Gaussian states lie on a manifold, as opposed to the discrete set of stabilizer states, we now have continuous degrees of freedom in the definition.
    $c$ is any complex-valued distribution such that $\ket{\Psi}$ is decomposed into Gaussian states. In particular, Dirac delta-like peaks are permitted, to give finite sums.
\end{defi}

Through standard methods (see the first Appendix of \cite{Heimendahl2021} for details) we can derive the dual form
\begin{alignat}{4}
    & \sqrt{\xi_G (\ket{\Psi})} = &&\max && \Re (\braket{\Psi}{y}) \label{eq:dual_problem} \\
    & && \text{s.t.} && \ket{y} \in \mathcal{H}_n &&\quad (\ket{y} \text{not necessarily normalised)} \nonumber \\
    & && && \sqrt{F_G(\ket{y})} \leq 1 \nonumber
\end{alignat}
Following standard theory, we define the \textit{feasible set}:
\begin{defi}
    The \textit{feasible set} $M_G$ for the dual problem \autoref{eq:dual_problem} is defined as
    $$
    M_G \coloneqq \{ \ket{y} \in \mathcal{H}_n : \abs{\braket{s}{y}}^2 \leq 1 \ \forall \, \ket{s} \in G \}
    $$
\end{defi}
Since $G$ contains a basis of $\mathcal{H}_n$, both problems are strictly feasible and strong duality holds. Also, since $\ket{\Psi} / \sqrt{F_G(\ket{\Psi})}$ is feasible for the dual, we get the lower bound
$$
\xi_G{(\ket{\Psi})} \geq \frac{1}{F_G(\ket{\Psi})}
$$

%%%%%%%%%%%%%%%%%%%%%%%%%%%%%%%%%%%%%%%%

\subsection{Optimal dual witnesses are generically unique, extreme points of the feasible set}
\label{subsec:unique_dual_witness}
Following the standard theory for second-order cone problems, we associate a \textit{normal cone} to each feasible point.
\begin{defi}
    The normal cone of a dually-feasible point $\ket{y} \in M_G$ is the set of primal points for which $\ket{y}$ is an optimal witness. In particular:
    $$
        C_y \coloneqq \left\{ 
        \ket{\Psi} \in \mathcal{H}_n : \Re(\braket{\Psi}{y}) = \max_{p \in M_G} \Re(\braket{\Psi}{p})
        \right\}
    $$
\end{defi}

The union over all $\ket{y} \in M_G$ of these cones must be all of $\mathcal{H}_n$, since every $\ket{\Psi} \in \mathcal{H}_n$ has an optimal dual witness.

It can easily be checked that an equivalent definition is
$$
    C_y = \text{cone}\left\{ \ket{s} \in G : \braket{s}{y} = 1 \right\}
$$
where the cone over a set $X$, denoted cone$\{X\}$, is the set of all linear combinations with non-negative coefficients of a finite subset of elements of $X$.

In~\autoref{sec:app.dual}, we show that if $\ket{\Psi} \in \text{relint}(C_y)$ where $\ket{y}$ is an extreme point of $M_G$ then $\ket{y}$ is the unique optimal dual witness. For a Haar-random point, this occurs almost-surely. For points $\ket{\Psi} \notin \text{relint}(C_y)$, optimal dual witnesses are not unique but may still be chosen to be extreme points of $M_G$.

\subsection{Towards a proof that Gaussian extent is not multiplicative in general}
To prove that extent is not multiplicative, we need provide only a single counterexample.
Our attempt to find one focused on using the methods of \cite{Heimendahl2021} in the Gaussian case.

In~\autoref{sec:fidelity}, we showed that  for a generic Haar-random $\ket{\Psi} \in \mathcal{H}_n$, $F_G(\ket{\Psi}) \leq 2^{-n + 3}n^4$.
In~\autoref{subsec:unique_dual_witness}, we showed that almost-surely, $\ket{\Psi}$ has a unique optimal dual witness $\ket{y}$. 
We may lower bound the norm of $\ket{y}$:
\begin{align} \label{eq:norm_y}
    \braket{y}{y} = \braket{y}{y} \braket{\Psi}{\Psi}
    \geq \abs{\braket{y}{\Psi}}^2 
    = \xi_G(\ket{\Psi}) 
    \geq \frac{1}{F_G(\ket{\Psi})} 
    \geq 2^{n-3}n^{-4}
\end{align}

Lemma 5 in \cite{Heimendahl2021} states that the extent is multiplicative on tensor product dictionaries.
The equivalent result here is:

\begin{lem}
    Let $\ket{\Psi_1}$ and $\ket{\Psi_2}$ be any states in $\mathcal{H}_n$ with optimal dual witnesses $\ket{y_1}$ and $\ket{y_2}$.
    Then $\xi_{G \otimes G}(\ket{\Psi_1}\ket{\Psi_2}) = \xi_G(\ket{\Psi_1}) \xi_G(\ket{\Psi_2})$ with optimal dual witness $\ket{y_1}\ket{y_2}$.
\end{lem}
\begin{proof}
    The proof is an immediate extension of the one in \cite{Heimendahl2021}.
\end{proof}
We therefore observe
\begin{equation} \label{eq:extent_ineq}
    \xi_{G_{2n}}(\ket{\Psi_1}\ket{\Psi_2}) \leq \xi_{G_n \otimes G_n}(\ket{\Psi_1}\ket{\Psi_2}) = \xi_{G_n}(\ket{\Psi_1})\xi_{G_n}(\ket{\Psi_2})
\end{equation}
The inequality in~\autoref{eq:extent_ineq} will become strict if there exists $\ket{s} \in G_{2n}$ with $\abs{\bra{s}(\ket{y_1}\ket{y_2})}^2 > 1$.
Since, by~\autoref{eq:norm_y}, $\norm{y_i}$ are exponentially large, it seems plausible that such a state exists.
However, it is not obvious how to find it. 
%%%%%%%%%%%%%%%%%%%%%%%%%%%%%%%%%%%%%%
\section{Gaussian rank} \label{sec:rank}
By the Gottesman-Knill Theorem, inner products of stabilizers with Clifford projectors $\bra{\phi_1}\Pi\ket{\phi_2}$ are classically efficiently computable. 
By decomposing the magic states injected into a Clifford circuit into a sum of stabilizer states, any quantum process may be brute-force simulated by computing inner products of this form \cite{Bravyi_2016}.

We define the minimum number of terms in the stabilizer decomposition of a state to be the Stabilizer Rank of that state~\cite{Bravyi_2016}. 
The number of terms in the above simulation method scales as $\chi(\ket{T}^{\otimes n})^2$, where $\chi(\ket{T}^{\otimes n})$ is the Stabilizer Rank of $n$ copies of the magic state $\ket{T}$. 
The stabilizer rank $\chi$ has received significant attention recently (e.g. \cite{Labib_2022, Peleg_2022, mehraban2023quadratic}). 
It is expected that $\chi_n$ must scale super-polynomially, else all quantum processes would be simulable in polynomial classical time.
However, it has proven difficult to achieve concrete results, and the state-of-the-art bound is currently $\chi(\ket{T}^{\otimes n}) = \Omega(n^2)$~\cite{mehraban2023quadratic}.

Since matchgate circuits are also classically simulable, it is natural to consider the same simulation technique in the MGC setting.

The starting point is to extend the set of matchgates to a gate set which allows for universal quantum computation.
One choice of such a gate is the SWAP gate~\cite{jozsa2008matchgates}.

Having expressed a general quantum process in terms of a ``matchgate + SWAP'' circuit, the next step is to replace all SWAP gates with the SWAP-gadget~\cite{Hebenstreit2019}.
These gadgets consist of an adaptive MGC acting on the lines we desire to SWAP, and consume $1$ copy of a magic state.

Hebenstreit et al~\cite{Hebenstreit2019} showed that all non-Gaussian states of definite parity are magic for MGCs.
We initially consider the $4$-qubit magic state $\ket{M} = \tfrac{1}{\sqrt{2}}(\ket{0} + \ket{15})$. 
Since this state violates the constraint $f(0,\,15) = 0$ given in~\autoref{eq:f_0_15}, it is not Gaussian, and is therefore magic. 
It also has the minimum possible Gaussian fidelity of $1/2$, and so in this sense it is intuitively ``maximally magic''.
Finally, this state is matchgate-equivalent to the magic state used in the SWAP-gadget of \cite{Hebenstreit2019}. 

\begin{defi}
    Let $\ket{\Psi}$ be an $n$-qubit state. The (exact) Gaussian rank $\chi_G(\ket{\Psi})$ is the smallest integer $k$ such that $\ket{\Psi}$ can be written as
    \begin{equation}
        \ket{\Psi} = \sum_{i = 1}^{k} c_i \ket{s_i}
    \end{equation}
    for some $n$-qubit Gaussian states $\ket{s_i}$ and some complex coefficients $c_i$.
\end{defi}

Clearly, $\ket{M}$ has Gaussian rank $\chi_G(\ket{M}) = 2$. Trivially, $\chi_G(\ket{M}^{\otimes k} \leq 2^k$.

We consider tensor products of $k = 2$ or 3 copies of $\ket{M}$ and seek out decompositions of rank less than $2^k$.

We also consider approximate decompositions of other choices of magic state which have a higher Gaussian fidelity and discuss the practical applicability of these choices.
One choice of magic state that offers a reasonable balance between non-Gaussianity and a relatively accurate decomposition is $\ket{\tilde{M}} \coloneqq \tfrac{1}{\sqrt{8}}(\sqrt{3}\ket{0} + \sqrt{2}\ket{3} + \sqrt{2}\ket{12} + \ket{15})$.

\subsection{Numerical results} \label{sec:numerical}
We used numerical optimization tools from the \textit{SciPy} \cite{2020SciPy} package to seek out low-rank decompositions.
The algorithm which found the most success was based on simulated annealing with cost function $L_{\ket{M}}(\ket{\Psi}) = \norm{\ket{{M}}^{\otimes k} - \ket{\Psi}}_2^2$.

Our program failed to find any exact low-rank decompositions for either 2 or 3 copies when using global optimization techniques for hundreds of hours of machine time. We view this as strong evidence that no such decompositions exist, but welcome future attempts with more refined techniques. 

Our code did see some success in finding approximate decompositions, but results were highly dependent on the choice of magic state. 
The code never found an approximate decomposition better than a trivial one for $\ket{M}\ket{M}$: using the 3 Gaussian states to get 3 out of 4 non-zero amplitudes correct.
Similarly, no good decomposition was found for $k = 3$.
We conjecture that good decompositions, even approximate ones, do not exist for only a few copies of this state due to the low Gaussian fidelity of $\ket{M}$.

Consider as a counterexample the magic state $\ket{M_\alpha} = \alpha \ket{0} + \sqrt{1-\alpha^2}\ket{15}$. 
As $\alpha \xrightarrow{} 1$, $F_G(\ket{M_\alpha}) \xrightarrow{} 1$ also.
For any $k$, it is then trivial to find an approximate Gaussian decomposition of $\ket{M_\alpha}^{\otimes k}$ with rank $1$; namely, $\ket{\Psi} = \alpha^k \ket{0}$. 
This decomposition has loss $L = \sqrt{2}k \alpha^{k-1} \sqrt{1-\alpha}$ to leading order in $(1-\alpha)$.
Such a magic state is not practically useful; we, therefore, choose to restrict to magic states which are bounded away from Gaussian states.

One alternative which did see success was the magic state $\ket*{\tilde{M}} = \tfrac{1}{\sqrt{8}}(\sqrt{3}\ket{0} + \sqrt{2}\ket{3} + \sqrt{2}\ket{12} + \ket{15})$.
For two copies of this state, decompositions with $L_{\ket*{\tilde{M}}}(\ket{\Psi}) \approx 10^{-5}$ were found.

Sample code is available, see~\autoref{sec:app.code}. 

\subsection{No symmetric low-rank decomposition for 2 copies of the magic state}
Based on the numerical results discussed in~\autoref{sec:numerical}, we conjecture that the Gaussian rank of 2 copies of $\ket{M}$ is $\chi_2(\ket{M}) = 4$.
We were unable to prove this in full generality, however, we were able to make some progress in a symmetry-reduced case.

Note that $\ket{M}$ is invariant under the (Gaussian) operation $Z_3 Z_4 = G_{34}(I,-I)$. 
This motivates considering a decomposition into Gaussian states which are invariant under the same operation on each block of $4$ qubits.

\begin{defi}
     Let $\ket{\Psi}$ be a $4n$-qubit state.
     For $j = 1,\,\ldots,\,n$, let $O_j \coloneqq Z_{4j-1} Z_{4j}$.    
     
     The \textit{symmetry-restricted Gaussian rank} $\tilde{\chi}_G(\ket{\Psi})$ is the smallest integer $k$ such that $\ket{\Psi}$ can be written as:
    \begin{equation}
        \ket{\Psi} = \sum_{i = 1}^{k} c_i \ket{s_i}
    \end{equation}
    for some complex coefficients $c_i$
    and some $n$-qubit Gaussian states $\ket{s_i}$ which satisfy $O_j\ket{s_i} = \ket{s_i}$ for $j = 1,\,\ldots,\,n$.
\end{defi}

\begin{prop} \label{prop:rank}
    The symmetry-restricted Gaussian rank of $\ket{M}\ket{M}$ is $\tilde{\chi}_G(\ket{M}\ket{M}) = 4$.
\end{prop}

We defer a full proof to~\autoref{sec:app.rank}, but give some discussion here.

The symmetry is chosen in order to greatly simplify the constraints on each of the Gaussian states.
Let $c_i \ket{s} = \sum_{x} a_x \ket{x}$ be a non-normalized Gaussian state with non-zero amplitudes only for computational-basis states $\ket{x}$ which are invariant under $O_j$. 
For $a_0 \neq 0$, we are left almost exclusively with constraints which only have 2 terms. 
We can visualise the constraints in a grid, wherein amplitudes in the body of the grid are defined in terms of the leading entries of their row and column.
\autoref{tab:constraints_0} represents all the constraints present on a restricted Gaussian state.

\begin{table}[t]
    \centering
    \begin{tabular}{c|ccccccc}
         0  & 3  & 65  & 66  & 129 & 130 & 192 & 195 \\\hline
         12 & 15 & 77  & 78  & 141 & 142 & 204 & 207 \\
         20 & 23 & 85  & 86  & 149 & 150 & 212 & 215 \\
         24 & 27 & 89  & 90  & 153 & 154 & 216 & 219 \\
         36 & 39 & 101 & 102 & 165 & 166 & 228 & 231 \\
         40 & 43 & 105 & 106 & 169 & 170 & 232 & 235 \\
         48 & 51 & 113 & 114 & 177 & 178 & 240 & 243 \\
         60 & 63 & 125 & 126 & 189 & 190 & 252 & 255 \\
    \end{tabular}
    \caption{Constraints on an 8-qubit, symmetry restricted Gaussian state, assuming $a_0 \neq 0$. The constraints are recovered by reading off a component in the body as the product of its row header and column header, divided by $a_0$. E.g. $a_{15} = a_3 a_{12} / a_0$, $a_{90} = a_{66} a_{24} / a_0$ etc.}
    \label{tab:constraints_0}
\end{table}

The only exceptions are a pair of constraints which have 4 terms in.
These can be viewed as fixing the values of the last entry in the first row and column respectively:
\begin{equation}
    \begin{aligned}
        a_0 a_{195} - a_3 a_{192} + a_{65} a_{130} - a_{66} a_{129} &= 0 \\
        a_0 a_{60} - a_{12} a_{48} + a_{20} a_{40} - a_{24} a_{36} &= 0
    \end{aligned}
\end{equation}

If $a_0 = 0$, then we must delete either the first row or the first column.
This is because we have from the second column of the grid:
$$
a_3 a_{12} = a_3 a_{20} = a_3 a_{24}= a_3 a_{36}= a_3 a_{40}= a_3 a_{48}= a_3 a_{60} = 0
$$
Thus either every term in the first column vanishes, or $a_3 = 0$. 
Proceeding similarly along the columns, we see that either the first column or first row must all vanish.

The table is still valid, however, as long as the value in the upper-leftmost entry is non-zero.
For example, we might have $a_{12} = a_{20} = a_{24}= a_{36}= a_{40}= a_{48}= a_{60} = 0$, $a_3 \neq 0$. Then equations are read off like $a_{77} = a_{65}a_{15}/a_3$.

The proof is essentially an exhaustive case bash. 
We use the fact that for $\ket{s_2}$, each term in the body of \autoref{tab:constraints_0} is  defined twice: 
once by the requirement that $\ket{s_2}$ is Gaussian; 
and a second time by the requirement that the $c_i \ket{s_i}$ sum to two copies of the magic state.
This gives a system of equations which we prove is insoluble.

Let $c_i \ket{s_i} = \sum_x a^i_x \ket{x}$. 
We first consider the case $a_0^i \neq 0$ for $i = 1,\,2,\,3$. 
This case allows for a somewhat neat matrix argument. Cases where some $a_0^i = 0$ become more complicated, spawning sub-cases as we must choose which row or column to delete.
Nonetheless, it is possible to find a contradiction in all cases. 

We note some obvious limitations of this argument.

Firstly, there is little justification for the symmetry reduction beyond some appeal to the symmetry of the problem. 
Without the reduction, the algebraic constraints are hard to work with, with no easy concept of independent vs. dependent components.
It therefore seems unlikely that our proof techniques would straightforwardly extend to the unrestricted case.

Also, the method is very much bespoke to the case of two copies of the magic state. 
With any more copies, the number of terms becomes too unwieldy and it is impractical to unwind all the definitions, even for the reduced case. We therefore expect that further results on lower bounds for the Gaussian rank will require new techniques that we have not considered here.

\section*{Acknowledgements}
SS acknowledges support from
the Royal Society University Research Fellowship and
“Quantum simulation algorithms for quantum chromodynamics” grant (ST/W006251/1) and EPSRC Reliable and Robust Quantum Computing grant (EP/W032635/1).
%%%%%%%%%%%%%%%%%%%%%%%%%%%%%%%%%%%%%%
\bibliography{gaussian_states_references}
\bibliographystyle{quantum}
%%%%%%%%%%%%%%%%%%%%%%%%%%%%%%%%%%%%%%
\appendix
%%%%%%%%%%%%%%%%%%%%%%%%%%%%%%%%%%%%%%%%
\section{Triples of Gaussian states} \label{sec:app.triple}
\begin{prop}
    The manifold of solutions to the system of equations 
    \begin{equation}
        \tilde{\mathcal{F}} = \left\{
        \sum_{i = 2}^{4} (-1)^i a_{\overline{0}^{k_1,k_i}} a_{\overline{y}^{k_1,k_i}} = 0 
        : \abs{x} = 4,\, D(0,\,x) = \{k_i : i = 1,\,\ldots,\,4\}
        \right\} \label{eq:constraint_sys}
    \end{equation}
    has dimension $2n -3$.
\end{prop}
For convenience, we introduce some new notation. For any $x \in \mathcal{A}_n$ with $\abs{x} = w$ and with 1s in locations $i_0,\,i_1,\ldots,\,i_{w-1}$, we will write $x = [i_0,\,i_1,\ldots,\,i_{w-1}]$. 
In expressions such as $[b_0,\,b_1,\,c_0,\,c_1]$ we will always have $b_0 < b_1$, $c_0 < c_1$ and $\{b_0,\,b_1\} \cap \{c_0,\,c_1\} = \emptyset$, but we may not know the relative ordering of the $b_i$ and $c_i$.
\begin{proof}
    Suppose $a_x \neq 0$ for some $x = [b_0,\,b_1]$, which we require for a non-trivial linearly dependent triple.
    
    Define $W_k = \{ y \in \mathcal{A}_n : \abs{y} = 4, \, d(x, y) = k \}$ for $k = 2,\,4,\,6$. 

    First, let $y \in W_2$. Then $y = [b_0,\,b_1,\,c_0,\,c_1]$. Since $a_x \neq 0$, we can write
    \begin{align}
        a_{[c_0,\, c_1]} &= \frac{1}{a_{[b_0,\,b_1]}} \bigl(
        s_0^{b_0,\,b_1,\,c_0,\,c_1} a_{[b_0,\,c_0]}a_{[b_1,\,c_1]} + s_1^{b_0,\,b_1,\,c_0,\,c_1} a_{[b_0,\,c_1]}a_{[b_1,\,c_0]}
        \bigr) \label{eq:fixed_var}\\
        s_0 &= \begin{cases}
            -1 \quad \text{if} \quad b_0<c_0<c_1<b_1 \quad \text{or} \quad c_0<b_0<b_1<c_1\\
            +1 \quad \text{otherwise}
        \end{cases} \nonumber \\
        s_1 &= \begin{cases}
            -1 \quad \text{if} \quad b_0<b_1<c_0<c_1 \quad \text{or} \quad c_0<c_1<b_0<b_1\\
            +1 \quad \text{otherwise}
        \end{cases} \nonumber
    \end{align}

  \autoref{eq:fixed_var} fixes the $\binom{n-2}{2}$ amplitudes $a_{z} = a_{[c_0,\,c_1]}$ with $d(x,\,z) = 4$. We now show that all other equations in $\tilde{\mathcal{F}}$ are identically satisfied. 
    Thus there are $\binom{n}{2} - \binom{n-2}{2} = 2n - 3$ free variables, giving the desired manifold dimension.

    Let $y \in W_4$. Then $y = [b_k,\,d_0,\,d_1,\,d_2]$ for $k = 0$ or 1. The corresponding equation in \autoref{eq:constraint_sys} reads
    \begin{align*}
        &t_0^{b_k, d_0,d_1,d_2} a_{[b_k, \, d_0]}a_{[d_1, \, d_2]} 
        + t_1^{b_k, d_0,d_1,d_2} a_{[b_k, \, d_1]}a_{[d_0, \, d_2]} 
        + t_2^{b_k, d_0,d_1,d_2} a_{[b_k, \, d_2]}a_{[d_0, \, d_1]} = 0\\
        &t_0 = \begin{cases}
            -1 \quad \text{if} \quad d_0<d_1<b_k<d_2\\
            +1 \quad \text{otherwise}
        \end{cases}\\
        &t_1 = \begin{cases}
            -1 \quad \text{if} \quad b_k<d_0<d_1<d_2 \quad \text{or} \quad d_0 < d_1<d_2<b_k\\
            +1 \quad \text{otherwise}
        \end{cases}\\
        &t_2 = \begin{cases}
            -1 \quad \text{if} \quad d_0<b_k<d_1<d_2\\
            +1 \quad \text{otherwise}
        \end{cases}
    \end{align*}
    
    Each of $a_{[d_1, \, d_2]}, a_{[d_0, \, d_2]} \,\&\, a_{[d_0, \, d_1]}$ are fixed by~\autoref{eq:fixed_var}. Substituting in, (multiplying out $a_{[b_0,\,b_1]}$), the equation now reads

    \begin{multline}
        t_0^{b_k, d_0,d_1,d_2}a_{[b_k, \, d_0]} \bigl(
        s_0^{b_0,b_1,d_1,d_2} a_{[b_0,\, d_1]}a_{[b_1,\, d_2]} + 
        s_1^{b_0,b_1,d_1,d_2} a_{[b_0,\, d_2]}a_{[b_1,\, d_1]}
        \bigr)\\
        + 
        t_1^{b_k, d_0,d_1,d_2}a_{[b_k, \, d_1]} \bigl(
        s_0^{b_0,b_1,d_0,d_2} a_{[b_0,\, d_0]}a_{[b_1,\, d_2]} + 
        s_1^{b_0,b_1,d_0,d_2} a_{[b_0,\, d_2]}a_{[b_1,\, d_0]}
        \bigr)\\
        +
        t_2^{b_k, d_0,d_1,d_2}a_{[b_k, \, d_0]} \bigl(
        s_0^{b_0,b_1,d_0,d_1} a_{[b_0,\, d_0]}a_{[b_1,\, d_1]} + 
        s_1^{b_0,b_1,d_0,d_1} a_{[b_0,\, d_1]}a_{[b_1,\, d_0]}
        \bigr) = 0 \label{eq:w4_equation}
    \end{multline}
    Labelling the terms of~\autoref{eq:w4_equation} in order 1 to 6, we can pair them off as follows. If $k = 0$, pair $(1,3),\,(2,5)\, \& \, (4,6)$. If $k=1$, instead pair $(1,6),\,(2,4)\,\&\,(3,5)$. Then it can be shown that each pair sums to 0.
    For example, for $k = 0$, we have:
    \begin{align*}
        (1) + (3) \quad &\propto \quad  
        t_0^{b_0, d_0,d_1,d_2}s_0^{b_0,b_1,d_1,d_2} 
        + t_1^{b_0, d_0,d_1,d_2}s_0^{b_0,b_1,d_0,d_2} \\
        (2) + (5) \quad &\propto \quad  
        t_0^{b_0, d_0,d_1,d_2}s_1^{b_0,b_1,d_1,d_2} + 
        t_2^{b_0, d_0,d_1,d_2}s_0^{b_0,b_1,d_0,d_1} \\
        (4) + (6) \quad &\propto \quad  
        t_1^{b_0, d_0,d_1,d_2}s_1^{b_0,b_1,d_0,d_2} 
        + t_2^{b_0, d_0,d_1,d_2}s_1^{b_0,b_1,d_0,d_1} 
    \end{align*}
    For example, note that $t_0$ and $t_1$ have the same sign only if $d_0 < b_0 < d_1 < d_2$. If this is the case, then $s_0^{b_0,b_1,d_1,d_2}$ and $s_0^{b_0,b_1,d_0,d_2}$ have different signs, as can be seen from checking the 3 possible relative locations of $b_1$. Thus $(1) + (3) = 0$. Similarly for the other pairs. The remaining cases can all be checked similarly.

    A similar argument holds for $y = [d_0,\,d_1,\,d_2,\,d_3] \in W_6$.
\end{proof}

%%%%%%%%%%%%%%%%%%%%%%%%%%%%%%%%%%%%%%%%%%
\section{\texorpdfstring{An $\epsilon$-net for the Gaussian states}{} \label{sec:app.net}}
For each $j \in \mathcal{A}_n$, define the regions
$$
S_j = \biggl\{ \ket{s} = \sum_{j' \in \mathcal{A}_n}c_{j'} \ket{j'} \in G: \abs{c_j} \geq \abs{c_{j'}} \quad \forall \, j' \in \mathcal{A}_n \biggr\}
$$
Clearly, the union over $j$ gives all of $G$. 
In each region $S_j$, Gaussian states may be defined by the set of amplitudes with labels within distance 2 of $j$, given by the set $C = \{ c_k : d(j, k) \leq 2 \}$.
For these amplitudes, we have 
\begin{align*}
    \abs{c_j} &\in [2^{\frac{1-n}{2}}, \, 1] \\
    \abs{c_k} &\in [0, \, \abs{c_j}] \quad \ \forall \, c_k \in C \setminus \{ c_j \}
\end{align*}
Let $\mathcal{N}_j = \{\ket{s_i} \} \subset S_j$ be a maximal set of pure states satisfying 
$$
\forall \, \ket{s_i}, \ket{s_{i'}} \in \mathcal{N}_j, \ \exists \ x \in \mathcal{A}_n \text{ s.t. } d(j, x) \leq 2  \quad \abs{\braket{x}{s_i} - \braket{x}{s_{i'}}} \geq \eta
$$
Let $\mathcal{N} = \cup_j \mathcal{N}_j$. We wish to find a suitable $\eta$ so that $\mathcal{N}$ is an $\epsilon$-net. We require the 2 following lemmas:

\begin{lem}\label{lem:eta_bound}
    For $\eta = 2^{-n^2} \epsilon $, we have 
    $\sup_{\ket{s} \in S_j} \inf_{\ket*{\tilde{s}} \in \mathcal{N}_j} \norm{\ket{s} - \ket*{\tilde{s}}}_1 \leq \epsilon \ \forall \, j \in \mathcal{A}_n$.
\end{lem}

\begin{lem}\label{lem:net_cardinality_bound}
    The cardinality of $\mathcal{N}_j$ is upper bounded by 
    $\abs{\mathcal{N}_j} \leq 2^{2n^2 + 1}\eta^{-n^2}$.
\end{lem}

\begin{proof}[Proof of \autoref{prop:eps_net}]
    By~\autoref{lem:eta_bound}, we have
    \begin{align*}
        \sup_{\ket{s} \in G} \inf_{\ket*{\tilde{s}} \in \mathcal{N}} \norm{\ket{s} - \ket*{\tilde{s}}}_1 
        &\leq \sup_{j \in \mathcal{A}_n} \sup_{\ket{s} \in S_j} \inf_{\ket*{\tilde{s}} \in \mathcal{N}_j} \norm{\ket{s} - \ket*{\tilde{s}}}_1 < \epsilon
    \end{align*}
    So $\mathcal{N}$ is a valid $\epsilon$-net.

    By~\autoref{lem:net_cardinality_bound} with $\eta = 2^{-n^2}\epsilon$, we have
    \begin{align*}
        \abs{\mathcal{N}} = \sum_j \abs{\mathcal{N}_j} \leq 2^{n-1} \cdot 2^{2n^2 +1} 2^{n^4}\epsilon^{-n^2}
    \end{align*}

    Let $\epsilon = 2^{-l}$. Then $\log_2 \abs{\mathcal{N}} \leq n^4 + 2n^2 + n+ ln^2$.
\end{proof}

\begin{proof}[Proof of~\autoref{lem:net_cardinality_bound}]
    Without loss of generality consider the region $S_0$ and sub-net $\mathcal{N}_0$.
    
    We use a volume argument to estimate the number of states in the sub-net. 
    Note that each state $\ket{s_i} \in \mathcal{N}_j$ occupies a disjoint region of $\binom{n}{2}$-dimensional (complex) space with volume $(\pi\eta^2)^{\binom{n}{2}}$.

    The centre of each of these disjoint regions lies in $S_0$. 
    The volume of the total region that they are contained in can be estimated by integrating over each free component $c_x$ with $\abs{x} \leq 2$. Viewed in $\mathbb{R}^2$,
    $c_0$ lies in an annulus with inner radius $2^{\frac{1-n}{2}} - \frac{\eta}{2}$ and outer radius $1+ \frac{\eta}{2}$. 
    Each $c_x$ lies in a disk of radius $\abs{c_0} + \frac{\eta}{2}$.
    Thus we see that the volume is no greater than
    \begin{align*}
        V &= 2\pi \int_{2^{\frac{1-n}{2}} - \frac{\eta}{2}}^{1 + \frac{\eta}{2}} dr_0 \left\{
        \left(\pi(r_0 + \frac{\eta}{2})^2\right)^{\frac{n(n-1)}{2}} 
        \right\}\\
        &= 2\pi^{\frac{n(n-1)}{2}+1} \left[
        \frac{(r_0 + \frac{\eta}{2})^{n(n-1)+1}}{n(n-1)+1} 
        \right]_{2^{\frac{1-n}{2}} - \frac{\eta}{2}}^{1 + \frac{\eta}{2}}\\
        &= 2\pi^{\frac{n(n-1)}{2}+1} \left(
        \frac{(1 + \eta)^{n(n-1)+1}}{n(n-1)+1} - \frac{(2^{\frac{1-n}{2}})^{n(n-1)+1}}{n(n-1)+1}
        \right)\\
        &\leq 2 \cdot \pi^{\frac{n^2}{2}} \cdot 2^{n(n-1)+1} \\
        &\leq 2 \cdot 2^{n^2} \cdot 2^{n^2 - n} = 2^{2n^2- n +1}
    \end{align*}

    The number of states is then no more than
    \begin{align*}
        \abs{\mathcal{N}_j} 
        &\leq \frac{V}{(\pi \eta^2)^{\binom{n}{2}}} \\
        &\leq 2^{2n^2 - n + 1}2^{- n(n-1)/2}\eta^{-n(n-1)} \\
        &\leq 2^{\frac{3}{2}n^2 - \frac{1}{2}n + 1}\eta^{-n^2} \\
        &\leq 2^{2n^2+1}\eta^{-n^2}\qedhere
    \end{align*}
\end{proof}

\begin{proof}[Proof of~\autoref{lem:eta_bound}]
    We will frequently use
    $$
    \abs{AB - A'B'} \leq \abs{A- A'}\abs{B} + \abs{B - B'}\abs{A'}
    $$
    and its generalisation to longer strings. Note also that all amplitudes have modulus $\leq 1$, so we can always take further
    $$
    \abs{AB - A'B'} \leq \abs{A- A'} + \abs{B - B'}
    $$
    when all the quantities involved are amplitudes.

    By definition, for any state $\ket{s} \in S_0$, $\exists \ket*{\tilde{s}} \in \mathcal{N}_j$ such that $\forall \, \abs{x} \leq 2$:
    \begin{gather*}
        c_x \coloneqq \braket{x}{s} \quad \tilde{c_x} \coloneqq \braket*{x}{\tilde s}\\
        \abs{c_x - \tilde{c_x}} \eqqcolon \abs{\delta_x} \leq \eta \eqqcolon \Delta_2 
    \end{gather*}
    For $\abs{x} \geq 4$, we find a recursion relation, bounding $\abs{\delta_x} \leq \Delta_w$ for $\abs{x} = w, \, w = 4,\, 6,\, 8, \dots, n$ (assume $n$ even for simplicity).
    We proceed by induction. Suppose we have bounded $\abs{\delta_x} \leq \Delta_{w'}$ for all $\abs{x} = w'$, $w' = 2,\, 4,\, \dots, w-2$. 
    
    Then for any $\abs{x} = w$, we have a constraint equation $c_0 c_x = f_x$, where $f_x$ contains $(w-1)$ terms of the form $c_{i}c_{i'}$ where $\abs{i} = 2$ and $\abs{i'} = w - 2$. Again using tildes to denote quantities relating to states on the net, we have
    \begin{align}
        \abs{\delta_x} &= \abs{c_x - \tilde{c}_x} \nonumber \\
        &= \abs{\frac{f_x}{c_0} - \frac{\tilde{f}_x}{\tilde{c}_0}} \nonumber \\
        &\leq \abs{c_0}^{-1} \abs{\tilde{c}_0}^{-1} \left( 
        \abs{\tilde{c}_0 - c_0} \abs{f_x} + \abs{f_x - \tilde{f}_x} \abs{\tilde{c}_0}
        \right)   \label{eq:delta_x}
    \end{align}
    We can bound each term in this expression. Assume that $\eta \leq 2^{\frac{1-n}{2}} - 2^{\frac{-(1+n)}{2}}$.
    Then we have:
    \begin{align*}
        \abs{c_0}^{-1} &\leq 2^{\frac{n-1}{2}} \\
        \abs{\tilde{c}_0}^{-1} &\leq (2^{\frac{1-n}{2}}- \eta)^{-1} \leq 2^{\frac{n+1}{2}} \\
        \abs{\tilde{c}_0 - c_0} &\leq \Delta_2 \\
        \abs{f_x} &= \abs{c_0} \abs{c_x} \leq 1 \\
        \abs{\tilde{c}_0} &\leq 1 \\
        \abs{f_x - \tilde{f}_x} &\leq (w-1)(\Delta_2 + \Delta_{w-2})
    \end{align*}
    Plugging in to~\autoref{eq:delta_x}, we obtain:
    \begin{align*}
        \abs{\delta_x} \leq 2^n (\Delta_2 + (w-1)(\Delta_2 + \Delta_{w-2}))
    \end{align*}
    Since this is true for any $x$ with $\abs{x} = w$, and noting that $\Delta_2 \leq \Delta_4 \ldots \leq \Delta_n$, we have
    $\Delta_w = 2^n (2w-1)\Delta_{w-2}$.
    The solution to this recurrence relation with $\Delta_2 = \eta$ is
    $$\Delta_w = \frac{1}{3}2^{\frac{1}{2}(n+2)w - n}\frac{\Gamma(\frac{w}{2} + \frac{3}{4})}{\Gamma(\frac{3}{4})} \eta $$.

    We then have
    \begin{align*}
        \norm{\ket{s} - \ket*{\tilde{s}}}_1
        &= \sum_{x \in \mathcal{A}_n} \abs{c_x - \tilde{c}_x} \\
        & \leq \sum_{k = 0}^{\frac{n}{2}}\binom{n}{2k} \Delta_{2k} \\
        & \leq 2^{n-1} \Delta_n \\
        & \leq \frac{1}{3} \cdot 2^{\frac{1}{2}n^2 + n - 1} \cdot \frac{\Gamma(\frac{n}{2}+\frac{3}{4})}{\Gamma(\frac{3}{4})} \eta
    \end{align*}
    Using $\Gamma(z+1) = z\Gamma(z)$, we may obtain 
    $$
    \frac{\Gamma(\frac{n}{2}+\frac{3}{4})}{\Gamma(\frac{3}{4})} 
    \leq \left(\frac{n}{2}\right)^{\frac{n}{2}} 
    \leq 2^{\frac{n}{2}(\log_2(n) - 1)}
    $$ 
    So finally we have 
    \begin{equation*}
        \norm{\ket{s} - \ket*{\tilde{s}}}_1 \leq 2^{\frac{1}{2}(n^2 + \log_2(n) + 1)}\eta \leq 2^{n^2}\eta
        \qedhere
    \end{equation*}
\end{proof}
%%%%%%%%%%%%%%%%%%%%%%%%%%%%%%%%%%%%%%%%%%%%%
\section{Properties of optimal dual witnesses for the Gaussian extent \label{sec:app.dual}}
Recall the equivalent definitions of normal cones given in the body:
\begin{align}
    C_y &\coloneqq \left\{ 
    \ket{\Psi} \in \mathcal{H}_n : \Re\braket{\Psi}{y} = \max_{p \in M_G} \Re\braket{\Psi}{p}
    \right\} \label{eq:cone_def_1}\\
    &= \text{cone}\left\{ \ket{s} \in G : \braket{s}{y} = 1 \right\} \label{eq:cone_def_2}
\end{align}

Suppose $\ket{y} \in \text{relint} (M_G)$. 
Then $\exists \, \alpha > 1$ s.t. $\alpha \ket{y} \in M_G$. 
Then $\forall \, \ket{\Psi} \in \mathcal{H}_n$:
$$
\Re\bra{ \Psi} \alpha \ket{y}
> \Re\braket{ \Psi}{ y } \implies C_y = \emptyset
$$
So we can restrict our attention only to points $y$ on the boundary $\partial M_G$.

Note that $M_G$ is a bounded, compact, convex set. The Krein-Milman Theorem allows us to express such sets in terms of convex combinations of their extreme points. 

\begin{thm}[Krein-Milman Theorem, see e.g. \cite{Rudin_1991}]
    Suppose $X$ is a compact, convex subset of a locally convex vector space. Then $X$ is equal to the closed convex hull of its extreme points. Moreover, for $B \subseteq X$, $X$ is equal to the closed convex hull of $B$ if and only if extreme$(X) \subseteq \text{closure}(B)$. \label{thm:km}
\end{thm}

Let $E$ denote the set of extreme points. Then for any $\ket{x} \in \partial M_G$, we can write
$$
\ket{x} = \sum_i c_i \ket{y_i}
$$
where $\ket{y_i} \in E$ and $c_i \in \mathbb{R}_+$ with $\sum_i c_i = 1$. Note that by Carath\'eodory's theorem, this sum can be taken to have only $d+1$ terms, and in particular is a finite sum, not an integral.

Let $\ket{\Psi} \in C_x$. Then
\begin{align*}
    \Re\braket{ \Psi}{  x }
    &= \Re\left( \sum_i c_i \braket{ \Psi}{  y_i } \right) \\
    &\leq \sum_i c_i \Re\braket{ \Psi}{ x } \\
    &= \Re\braket{ \Psi}{  x }
\end{align*}
where on the second line we used that the $c_i$ are real and also that $\ket{\Psi} \in C_x$, and in the third line we used that the $c_i$ sum to unity.

We must therefore have equality on the second line, so $\Re\braket{ \Psi}{ x } = \Re\braket{ \Psi}{  y_i }$. 
By~\autoref{eq:cone_def_1}, $\ket{x}$ maximises $\Re\braket{\Psi}{p}$ over feasible $\ket{p}$. 
For each $i$, $\ket{y_i}$ is feasible, and also attains this maximum. 
Therefore, $\ket{\Psi} \in C_{y_i}$.

Therefore, $\ket{\Psi} \in \mathcal{H}_n$ has non-unique optimal dual witnesses iff $\ket{\Psi}$ lies in the intersection of some $C_{y_i}$ for $\ket{y_i} \in E$. We now show that the $C_{y_i}$ intersect only on their boundaries.

\begin{prop} \label{prop:unique_cone}
    For distinct $\ket{y}, \, \ket{y'} \in E$,
    $\ket{\Psi} \in \text{relint}(C_y) \implies \ket{\Psi} \notin C_{y'}$.
\end{prop}

Before proceeding with the proof, we will need some technical Lemmas:

\begin{lem}\label{lem:extreme_span}
    Let $\ket{y} \in M_G$. 
    Define the sets:
    $$
    A_y = \left\{ \ket{s} \in G : \braket{s}{y} = 1 \right\} \quad \overline{A}_y = \left\{ \ket{t} = U X_i X_j U^\dag \ket{s} : \ket{s} = U \ket{0} \in A_y ; i \neq j \right\}
    $$
    Then $\ket{y} \in E \implies \text{span}(A_y \cup \overline{A}_y) = \mathcal{H}_n$.
\end{lem}

$A_y$ is the set of states which form a basis for $C_y$ as defined by~\autoref{eq:cone_def_2}.
We now motivate the rather strange definition $\overline{A}_y$. 

Hamming distance plays a key role in the definition of Gaussian states; computational basis states with distance at least 4 are constrained relative to one another.
Thus $\overline{A}_y$ in some sense encodes states $\ket{t}$ which are unconstrained with respect to some $\ket{s} \in A_y$. By this we mean that $\sqrt{\alpha} \ket{s} + \sqrt{1 -\alpha} \ket{t} \in G$ for any $0 \leq \alpha \leq 1$.

Conversely, any $\ket{u} \in \text{span}(A_y \cup \overline{A}_y)^\perp$ has a non-trivial constraint with respect to every $\ket{s} \in A_y$. The best we can do is $\alpha \ket{s} + (1-\alpha)\ket{t} + \{\text{balancing terms}\}$. 

In particular, we have $\abs{a_s}^2 + \abs{a_u}^2 < 1$ where $a_s$ is the part of a Gaussian state in $\text{span}(A_y)$ and $a_u$ the part in $ \text{span}(A_y \cup \overline{A}_y)^\perp$.

\begin{proof}
    Suppose $A_y \cup \overline{A}_y$ does \textit{not} span $\mathcal{H}_n$. 
    
    Let $S = \left\{ \ket{s_i} \right\}$ be a basis for $A_y$, $T = \left\{ \ket{t_i} \right\}$ a basis for $\overline{A}_y$ and $U = \left\{ \ket{u_i} \right\}$ a basis for $\text{span}(A_y \cup \overline{A}_y)^\perp$.

    For any $\ket{t} \in \overline{A}_y$, with corresponding element $\ket{s} \in A_y$, note that 
    \begin{align*}
        \ket{\psi} &\coloneqq \sqrt{\alpha}\ket{s} + \sqrt{1 - \alpha}e^{-i\phi} \ket{t} \\
        &= U\left(\sqrt{\alpha}\ket{0} + \sqrt{1 - \alpha}e^{-i\phi} X_i X_j\ket{0}\right) \in G
    \end{align*} 
    for any $0 \leq \alpha \leq 1,\, 0 \leq \phi < 2\pi$. 
    
    Then $\ket{y} \in M_G$ requires $\abs{\braket{\psi}{y}}^2 = \abs{\sqrt{\alpha} + \sqrt{1-\alpha}e^{-i\phi}\braket{t}{y}}^2 \leq 1$. 
    Taking $\alpha = 1/(1 + \abs{\braket{t}{y}}^2)$ and $\phi =\text{phase}(\braket{t}{y})$ we see that $\braket{t}{y} = 0$.
    
    Consider $\ket{y_\pm} = \ket{y} \pm \epsilon \ket{u_1}$. We show that these points are feasible for some $\epsilon > 0$, and so $\ket{y}$ is a proper convex combination of feasible points.
    
    First, for any $\ket{\psi} \in G$, first note that $\braket{\psi}{u_1} = 0 \implies \abs{\braket{\psi}{y_\pm}}^2 = \abs{\braket{\psi}{y}}^2 \leq 1$.

    For $\braket{\psi}{u_1} \neq 0$, write $\ket{\psi} = a_s \ket{s} + a_t \ket{t} + a_u \ket{u}$ where $\ket{s} \in \text{span}(A_y)$ etc. Then, noting $\braket{s}{y} = 1,\, \braket{t}{y} = 0$ and $\braket{s}{u} = \braket{t}{u} = 0$, we have:
    \begin{align*}
        \abs{\braket{\psi}{y_\pm}}^2 &= \abs{a_s + a_u \braket{u}{y} \pm \epsilon a_u \braket{u}{u_1}}^2\\
        &\leq \abs{a_s}^2 + \abs{a_u}^2(\abs{\braket{u}{y}}^2 + \epsilon^2 \abs{\braket{u}{u_1}}^2)
    \end{align*}

    If we also write $\ket{y} = b_{s'} \ket{s'} + b_{t'}\ket{t'} + b_{u'} \ket{u'}$, it is clear that $\sup_{\ket{u}}(\abs{\braket{u}{y}}^2) = \abs{b_{u'}}^2$. Moreover, $\abs{b_{u'}} < 1$ else $\braket{u'}{y} = 1$, implying $\ket{u'} \in A_y \cap \text{span}(A_y \cup \overline{A}_y)^\perp$, which is impossible.

    Noting $\abs{\braket{u}{u_1}}^2 \leq 1$, we can therefore take $0 < \epsilon < \sqrt{1 - \abs{b_u}^2}$ to obtain
    $$ \abs{\braket{\psi}{y_\pm}}^2 < \abs{a_s}^2 + \abs{a_u}^2 < 1$$
    So we have $F_G(\ket{y_\pm}) \leq 1$. Thus these are feasible points, and the proper convex combination $\ket{y} = \frac{1}{2}(\ket{y_+} + \ket{y_-})$ implies $\ket{y} \notin E$. 
\end{proof}

\begin{lem} \label{lem:cone_strict_subset}
    For distinct $\ket{y}, \, \ket{y'} \in E$,
    $C_y \not\subset C_{y'}$, where the subset is strict.
\end{lem}

\begin{proof}    
    Define bases for the sets $A_y$ and $A_{y'}$:
    $$
    S_y \coloneqq \left\{ \ket{s_i} : i = 1,\ldots,m \right\} \subset \left\{ \ket{s_i} : i = 1,\ldots,M \right\} \eqqcolon S_y'
    $$
    where $M > m$. Since $\ket{y} \in E$, we must have $\text{span}(A_y)\, \oplus \, \text{span}(A_{y'}) = \text{span}(A_y \cup \overline{A}_y) = \mathcal{H}_n$.

    Consider $\ket{s_M} \in S_{y'} \setminus S_y$. We may write $\ket{s_M} = \ket{s} + \ket{t}$ for $\ket{s} \in \text{span}(A_y),\, \ket{t} \in \text{span}(\overline{A}_y)$. As in the proof of~\autoref{lem:extreme_span}, we have that $\braket{t}{y} = 0$. Noting that $\overline{A}_y \subseteq \overline{A}_{y'}$, we also have $\braket{t}{y'} = 0$. Also, since $A_y \subset A_{y'}$, we have $\braket{y}{s} = \braket{y'}{s}$.

    Then since $\ket{s_M} \in A_{y'}$, $1 = \braket{y'}{s_M} = \braket{y'}{s} + \braket{y'}{t} = \braket{y'}{s}$. But then $\braket{y}{s_M} = \braket{y}{s} + \braket{y}{t} = \braket{y}{s} = \braket{y'}{s} = 1$, which implies $\ket{s_M} \in A_y$, which is a contradiction.
\end{proof}

\begin{lem}\label{lem:cone_equal}
    For distinct $\ket{y},\,\ket{y'} \in E$, $C_y \neq C_{y'}$.
\end{lem}

\begin{proof}
    We have $\text{span}(A_y)\oplus \,\text{span}(\overline{A}_y) = \mathcal{H}_n$. We also have $1 = \braket{s}{y} \ \forall \, \ket{s} \in A_y$ and $0 = \braket{t}{y} \ \forall \, \ket{t} \in \overline{A}_y$. This system of equations uniquely determines $\ket{y}$. In particular:
    \begin{equation*}
         C_y = C_{y'} \implies A_y = A_{y'} \implies \overline{A}_y = \overline{A}_{y'} \implies y = y' \qedhere 
    \end{equation*}
\end{proof}

\begin{proof}[Proof of \autoref{prop:unique_cone}]
    Let $\ket{\Psi} \in C_y$.
    By \autoref{lem:cone_strict_subset} \& \autoref{lem:cone_equal}, we have $C_y \nsubseteq C_{y'}$. So $\exists \, \ket{z} \in C_y \setminus C_{y'}$. Then by the definition of relative interior for convex sets, we have
    $$
        \exists \, \ket{w} \in C_y, \, 0 < \lambda < 1 \quad \text{s.t.} \quad \ket{\Psi} = \lambda \ket{w} + (1-\lambda)\ket{z}
    $$
    We then have
    \begin{align*}
        \braket{ \Psi}{  y' }^R
        &= \lambda \braket{ w}{ y' }^R 
        + (1-\lambda) \braket{ z}{  y' }^R \\
        &< \lambda \braket{ w}{  y }^R 
        + (1-\lambda) \braket{ z}{  y }^R \\
        &= \braket{ \Psi}{  y }^R
    \end{align*}
    where the second line uses the facts
    \begin{align*}
        \braket{ z}{ y' }^R 
        &< \braket{ z}{ y }^R \\
        \braket{ w}{  y' }^R
        &\leq \braket{ w}{ y }^R
    \end{align*}
    since $\ket{w} \in C_y,\, \ket{z} \in  C_y \setminus C_{y'}$.
    Therefore, $\ket{\Psi} \notin C_{y'}$.
\end{proof}

%%%%%%%%%%%%%%%%%%%%%%%%%%%%%%%%%%%%%%%%%%%%%
\section{Symmetry-restricted Gaussian rank of 2 copies of a magic state}\label{sec:app.rank}

We give a sketch proof of \autoref{prop:rank} to give a flavour of our proof techniques.
We consider several cases, finding a contradiction in each case.
Recall that the equation we seek to solve is
\begin{equation} \label{eq:gaussian_decomp}
    \ket{M}\ket{M} = \frac{1}{2}(
        \ket{0} + \ket{15} + \ket{240} + \ket{255}
    )
    = \sum_{i = 1}^3 \sum_x a^i_x \ket{x}
\end{equation}
where the $\ket{x}$ are invariant under $Z_3 Z_4$ and $Z_7 Z_8$.

\subsection{\texorpdfstring{$a_0^i \neq 0,\,i=1,\,2,\,3$}{}}
Consider as illustrative examples the two equations each which are enforced on $\aa{15}{3}$ and $\aa{51}{3}$:
\begin{align}
    \aa{15}{3} = (\frac{1}{2}-\aa{15}{1}-\aa{15}{2}) &= \frac{\aa{12}{3} \aa{3}{3}}{\aa{0}{3}} \nonumber\\
    \implies \frac{1}{2}\aa{0}{3} &= 
    \frac{\aa{3}{1}}{\aa{0}{1}}(\frac{1}{2}\aa{12}{1} - \aa{12}{1}\aa{0}{2} + \aa{12}{2}\aa{0}{1}) +
    \frac{\aa{3}{2}}{\aa{0}{2}}(\frac{1}{2}\aa{12}{2} - \aa{12}{2}\aa{0}{1} + \aa{12}{1}\aa{0}{2}) \label{eq:a_15}\\
    \vspace{12pt} \nonumber
    \\
    \aa{51}{3} = -(\aa{51}{1}+\aa{51}{2}) &= \frac{\aa{48}{3} \aa{3}{3}}{\aa{0}{3}} \nonumber\\
    \implies 0 &= 
    \frac{\aa{3}{1}}{\aa{0}{1}}(\frac{1}{2}\aa{48}{1} - \aa{48}{1}\aa{0}{2} + \aa{48}{2}\aa{0}{1}) +
    \frac{\aa{3}{2}}{\aa{0}{2}}(\frac{1}{2}\aa{48}{2} - \aa{48}{2}\aa{0}{1} + \aa{48}{1}\aa{0}{2}) \label{eq:a_51}
\end{align}
\autoref{eq:a_15} \& \autoref{eq:a_51} are structurally very similar. 
Define a function which captures the terms in brackets:
$$
f^{i}_{j} \coloneqq \frac{1}{2}\aa{j}{i} +(-1)^i (\aa{0}{1}\aa{j}{1} - \aa{0}{2}\aa{j}{0})
$$
Then proceeding similarly for each other term in the body of~\autoref{tab:constraints_0}, we obtain the matrix equation:
\begin{gather}
    \begin{pmatrix}
    A \\ A'
    \end{pmatrix}
    \begin{pmatrix}
    X & X'
    \end{pmatrix}
    =
    \begin{pmatrix}
    \frac{1}{2}\aa{0}{3} I_3 & 0 \\
    0 & 0
    \end{pmatrix} \label{eq:matrix_0}
\end{gather}
\begin{alignat*}{2}
    A & = \begin{pmatrix}
        f^0_{12} & f^1_{12} \\
        f^0_{48} & f^1_{48} \\
        f^0_{60} & f^1_{60}
    \end{pmatrix} \quad
    & A' &= \begin{pmatrix}
        f^0_{20} & f^1_{20} \\
        f^0_{24} & f^1_{24} \\
        f^0_{36} & f^1_{36} \\
        f^0_{40} & f^1_{40} 
    \end{pmatrix}
    \\
    X &= \begin{pmatrix}
        \frac{\aa{3}{1}}{\aa{0}{1}} & \frac{\aa{192}{1}}{\aa{0}{1}} & \frac{\aa{195}{1}}{\aa{0}{1}} \\
        \frac{\aa{3}{2}}{\aa{0}{2}} & \frac{\aa{192}{2}}{\aa{0}{2}} & \frac{\aa{195}{2}}{\aa{0}{2}}
    \end{pmatrix} \quad
    & X' &= \begin{pmatrix}
        \frac{\aa{65}{1}}{\aa{0}{1}} & \frac{\aa{66}{1}}{\aa{0}{1}} & \frac{\aa{129}{1}}{\aa{0}{1}} & \frac{\aa{130}{1}}{\aa{0}{1}} \\
        \frac{\aa{65}{2}}{\aa{0}{2}} & \frac{\aa{66}{2}}{\aa{0}{2}} & \frac{\aa{129}{2}}{\aa{0}{2}} & \frac{\aa{130}{1}}{\aa{0}{1}}
    \end{pmatrix}
\end{alignat*}
Note that
$\text{rank}(AX) \leq \min(\text{rank}(A),\text{rank}(X))\leq 2$,
whereas 
$\text{rank}(\aa{0}{3} I_3) = 3$
since by assumption $\aa{0}{3} \neq 0$. Thus there is no solution to \autoref{eq:matrix_0}.

\subsection{\texorpdfstring{$a_0^2 = 0, \ a_0^i \neq 0$, $i = 0, \, 1$}{}}
By symmetry, without loss of generality choose to delete the first column of \autoref{tab:constraints_0} for $\ket{\psi_3}$, i.e. 
$$\aa{12}{3} = \aa{20}{3} = \aa{24}{3} = \aa{36}{3} = \aa{40}{3} =\aa{48}{3} = \aa{60}{3} = 0$$
Then also note that, by~\autoref{eq:gaussian_decomp},
$
    \aa{j}{2} = - \aa{j}{1} \ \ j = 12, \,20, \,24, \,36, \,40, \,48, \,60
$.

Then our special-case constraint gives:
\begin{equation} \label{eq:a_60}
    \aa{60}{1} = \frac{\aa{12}{1}\aa{48}{1} - \aa{20}{1}\aa{40}{1} + \aa{24}{1}\aa{36}{1}}{\aa{0}{1}} =
    - \frac{\aa{12}{2}\aa{48}{2} - \aa{20}{2}\aa{40}{2} + \aa{24}{2}\aa{36}{2}}{\aa{0}{2}} = -\aa{60}{2}
\end{equation}
Then either $\aa{0}{2} = -\aa{0}{1}$ or $\aa{60}{1} = \aa{60}{2} = 0$.

In the first case, we immediately have $\sum_i \aa{0}{i}  =0$, which is not possible by~\autoref{eq:gaussian_decomp}.

In the second case, since $\aa{0}{1},\,\aa{0}{2} \neq 0$, we have 
\begin{align*}
    \implies \aa{j}{1} = \aa{j}{2} = 0, \, j &= 63, 125, \, 126, \, 189, \, 190,\, 252, \, 255 \ \  
    (\text{since} \ \aa{j}{1} \propto \aa{60}{1}, \aa{j}{2} \propto \aa{60}{2}) \\
    \implies \aa{j}{3} = 0, \, j &= 63, 125, \, 126, \, 189, \, 190,\, 252 \ \ (\text{since} \ \sum_i\aa{j}{i} = 0)
\end{align*}
Consider $\aa{255}{3}$. Since $\aa{255}{1} = \aa{255}{2} = 0$ we must have $\aa{255}{3} \neq 0$.
Since the rest of the row of 255 is 0 for $\ket{\psi_3}$, it must be the case that only the last column is non-zero for $\ket{\psi_3}$.

Finally, note that 
\begin{align*}
    \sum_i \aa{15}{i} = \aa{12}{0}\left(
        \frac{\aa{3}{0}}{\aa{0}{0}} 
        - \frac{\aa{3}{1}}{\aa{0}{1}}
    \right) \neq 0 
    \\
    \sum_i \aa{51}{i} = \aa{48}{0}\left(
        \frac{\aa{3}{0}}{\aa{0}{0}} 
        - \frac{\aa{3}{1}}{\aa{0}{1}}
    \right) = 0
\end{align*}
Clearly we must have $\aa{48}{0} = \aa{48}{1} = 0 $.
But then $\aa{240}{0} = \aa{240}{1} = 0$, so $\sum_i \aa{240}{i} = 0$, contradicting~\autoref{eq:gaussian_decomp}.
\subsection{\texorpdfstring{$a_0^1 = a_0^2 = 0, \ a_0^0 \neq 0$}{}}
This case spawns several sub-cases depending on which column of \autoref{tab:constraints_0} we delete for each of $\ket{\psi_2}$ \& $\ket{\psi_3}$. 

In each case, the reasoning follows much as the above, and we exclude the details.
%%%%%%%%%%%%%%%%%%%%%%%%%%%%%%%%%%%%%%%%%%%%%%%
\section{Code repository} \label{sec:app.code}
Code is available on GitHub \href{https://github.com/JoshCudby/GaussianDecomposition}{here}.
\end{document}